\documentclass[11pt]{paper}

\usepackage{latexsym}
\usepackage{fullpage}
\usepackage{verbatim}
\usepackage{cite}
\usepackage{amssymb}
\usepackage{amsmath}
\usepackage{amsthm}
\usepackage{graphicx}
\usepackage{color}
\usepackage{phaistos}
\usepackage[disable]{todonotes}
\usepackage{wrapfig}
\usepackage{stfloats}


\newtheorem{theorem}{Theorem}[section]
\newtheorem{lemma}[theorem]{Lemma}
\newtheorem{conj}[theorem]{Conjecture}
\newtheorem{claim}[theorem]{Claim}

\newcommand{\eqdef}{:=}
\newcommand{\eps}{\varepsilon}
\DeclareMathOperator{\FSD}{FSD}
\DeclareMathOperator{\reach}{reach}

\usepackage{xspace}

\newcommand\frechet{Fr\'echet\xspace}
\newcommand\etal{\emph{et~al.}}

\newcommand\N{\mathbb{N}}

\def\bd{{\partial}}

\newcommand\R{{\mathbb R}}

\graphicspath{ {./figures/} }

\begin{document}

\title{Four Soviets Walk the Dog---Improved Bounds
for Computing the \frechet Distance\footnote{A preliminary version 
appeared as
K.\@ Buchin, M.\@ Buchin, W.\@ Meulemans, and W.\@ Mulzer.
\emph{Four Soviets Walk the Dog---with an Application to Alt's Conjecture}.
Proc.~25th SODA, pp.~1399--1413, 2014.
K. Buchin and M. Buchin in part supported by 
COST (European Cooperation in Science and Technology) 
ICT Action IC0903 MOVE.
M. Buchin supported in part by
the Netherlands Organisation for Scientific Research
(NWO) under project no. 612.001.106.
W. Meulemans supported by the Netherlands
Organisation for Scientific Research (NWO) under project no.~639.022.707.
W. Mulzer supported in part by DFG projects MU/3501/1 and MU/3501/2.
}
}
\author{
Kevin Buchin\thanks{Dept.~of Mathematics and Computer Science,
TU Eindhoven, The Netherlands,
\texttt{[k.a.buchin, w.meulemans]@tue.nl}.
}
\and
Maike Buchin\thanks{Fakult\"at f\"ur Mathemtaik,
Ruhr Universit\"at Bochum, Germany,
\texttt{Maike.Buchin@ruhr-uni-bochum.de}.
}
\and
Wouter Meulemans\footnotemark[2]
\and
Wolfgang Mulzer\thanks{
Institut f\"ur Informatik,
Freie Universit\"at Berlin, Germany, \texttt{mulzer@inf.fu-berlin.de}.
}
}

\date{}

\maketitle

\begin{abstract}
Given two polygonal curves in the 
plane, there are many ways to define 
a notion of similarity between them. 
One popular measure 
is the \frechet distance.  Since it 
was proposed by Alt and Godau in 1992, 
many variants and extensions have been 
studied. Nonetheless, even more than 
20 years later, the original $O(n^2 \log n)$ 
algorithm by Alt and Godau for computing 
the \frechet distance remains the 
state of the art (here, $n$ denotes 
the number of edges on each curve). 
This has led Helmut Alt to conjecture 
that the associated decision problem
is 3SUM-hard.

In recent work, Agarwal \etal~show how 
to break the quadratic barrier for the 
\emph{discrete} version of the \frechet 
distance, where one considers sequences 
of points instead of polygonal curves. 
Building on their work, we give a 
randomized algorithm to compute the 
\frechet distance between two polygonal 
curves in time $O(n^2 \sqrt{\log n}(\log\log n)^{3/2})$ 
on a pointer machine and in time $O(n^2(\log\log n)^2)$ 
on a word RAM. Furthermore, we show that 
there exists an algebraic decision tree 
for the decision problem of depth $O(n^{2-\eps})$, 
for some $\eps > 0$. We believe that this 
reveals an intriguing new aspect of this 
well-studied problem.
Finally, we show how 
to obtain the first subquadratic algorithm
for computing the weak \frechet distance on a
word RAM.
\end{abstract}

\section{Introduction}

Shape matching is a fundamental problem in 
computational geometry, computer vision, and 
image processing.  A simple version can be 
stated as follows: given a database $\mathcal{D}$ 
of shapes (or images) and a query shape $S$, 
find the shape in $\mathcal{D}$ that most 
resembles $S$. However, before we can solve 
this problem, we first need to address an 
issue: what does it mean for 
two shapes to be ``similar''?  In the 
mathematical literature, on can find many 
different notions of distance between two sets, 
a prominent example being the \emph{Hausdorff 
distance}. Informally, the Hausdorff distance 
is defined as the maximal distance between 
two elements when every element of one set is 
mapped to the closest element in the other.
It has the advantage of being simple to 
describe and easy to compute for discrete sets.
In the context of shape matching, however, the 
Hausdorff distance often turns out to be 
unsatisfactory: it does not take the continuity 
of the shapes into account.
There are well known examples where the distance 
fails to capture the similarity of shapes as 
perceived by human observers~\cite{Alt09}.

In order to address this issue, Alt and Godau 
introduced the \frechet distance into the 
computational geometry literature~\cite{Godau91a,AltGo95}.
They argued that the \frechet distance is 
better suited as a similarity measure, and 
they described an $O(n^2\log n)$ time algorithm 
to compute it on a real RAM or pointer 
machine.\footnote{For a brief overview of the 
different computational models in this paper, 
refer to Appendix~\ref{app:compmodels}.}
Since Alt and Godau's seminal paper, there has 
been a wealth of research in various directions, 
such as extensions to higher dimensions~\cite{AltBuchin10,
Buchin07,BuchinBuSc10,BuchinBW08,CookDHSW11,Godau98}, 
approximation algorithms~\cite{AltKW03,AronovHKWW06,DriemelHW12}, 
the geodesic and the homotopic \frechet 
distance~\cite{ChambersVELLT10,CookWenk10,EfratGHMM02,HarPeledNSS12}, 
and much more~\cite{AgarwalHMS05,BergCG13,BuchinBMS12,BuchinBW09,
DriemelHarpeled13,Indyk02,MaheshwariSSZ11,MaheshwariSSZ11b}.
Most known approximation algorithms make further 
assumptions on the curves, and only an $O(n^2)$-time 
approximation algorithm is known for arbitrary 
polygonal curves~\cite{bblmm-fdrl-16}.
The \frechet distance and its variants, such as 
\emph{dynamic time-warping}~\cite{BellmanKalaba59}, 
have found various applications, with recent work 
particularly focusing on geographic applications 
such as map-matching tracking data~\cite{BrakatsoulasPSW05,WenkSP06} 
and moving objects analysis~\cite{BuchinBG10,
BuchinBGLL11,GudmundssonWolle10}.

Despite the large amount of published research, 
the original algorithm by Alt and Godau has not 
been improved, and the quadratic barrier on 
the running time of the associated decision 
problem remains unbroken. If we cannot improve 
on a quadratic bound for a geometric problem 
despite many efforts, a possible culprit may 
be the underlying 3SUM-hardness~\cite{GajentaanOv95}.
This situation induced Helmut Alt to make the 
following conjecture.\footnote{Personal 
communication 2012, see also~\cite{Alt09}.}

\begin{conj}[\textsc{Alt's Conjecture}]
Let $P$, $Q$ be two polygonal curves in the
plane. Then it is \textup{3SUM}-hard to 
decide whether the \frechet distance
between $P$ and $Q$ is at most $1$.
\end{conj}

Here, $1$ can be considered as an arbitrary 
constant, which can be changed to any other 
bound by scaling the curves. So far, the 
best unconditional lower bound for the problem
is $\Omega (n \log n)$ steps in the algebraic 
computation tree model~\cite{BuchinBKRW07}.

Recently, Agarwal~\etal~\cite{AgarwalBAKaSh14}
showed how to achieve a subquadratic running 
time for the \emph{discrete} version of the
\frechet distance, running in 
$O\left(n^2 \, \frac{\log\log n}{\log n}\right)$ time.
Their approach relies on 
reusing  small parts of the solution. We 
follow a similar approach based on the
so-called \emph{Four-Russian-trick} which 
precomputes small recurring parts of the 
solution and uses table-lookup to speed up 
the whole computation.\footnote{It is well 
known that the four Russians are not actually 
Russian, so we refer to them as four Soviets 
in the title.}
The result by Agarwal~\etal~is stated in the 
word RAM model of computation. They ask whether
their result can be generalized to the case of 
the original (continuous) \frechet distance.

\paragraph{Our contribution}
We address the question by Agarwal~\etal~and show 
how to extend their approach to the \frechet 
distance between two polygonal curves. Our 
algorithm requires total expected time
$O(n^2 \sqrt{\log n} (\log\log n)^{3/2})$.
This is the first algorithm with a running 
time of $o(n^2 \log n)$ and constitutes the 
first improvement for the general case since 
the original paper by Alt and Godau~\cite{AltGo95}. 
To achieve this running time, we give the first 
subquadratic algorithm for the decision problem 
of the \frechet distance. We emphasize that 
these algorithms run on a real RAM/pointer machine 
and do not require any bit-manipulation tricks.
Therefore, our results are more in the line of 
Chan's recent subcubic-time algorithms for 
all-pairs-shortest paths~\cite{Chan08,Chan10} 
or recent subquadratic-time algorithms for 
min-plus convolution~\cite{BremnerChDeErHuIaLaTa12} 
than the subquadratic-time algorithms for 3SUM
due to Baran \etal~\cite{BaranDP08}.  
If we relax the model to allow constant time 
table-lookups, the running time can be improved 
to be almost quadratic, up to $O(\log\log n)$ 
factors.  As in Agarwal~\etal, our results are 
achieved by first giving a faster algorithm 
for the decision version, and then performing 
an appropriate search over the critical values 
to solve the optimization problem.

In addition, we show that \emph{non-uniformly}, 
the \frechet distance can be computed in subquadratic 
time. More precisely, we prove that the decision 
version of the problem can be solved by
an algebraic decision tree~\cite{AroraBa09} of 
depth $O(n^{2-\eps})$, for some fixed
$\eps > 0$. It is, however, not clear how 
to implement this decision tree in subquadratic
time, which hints at a discrepancy between the 
decision tree and the uniform complexity of the 
\frechet problem.

Finally, we consider the \emph{weak} \frechet distance,
where we are allowed to walk backwards along the curves.
In this case, our framework allows us to achieve
a subquadratic algorithm on the work RAM.
Refer to Table~\ref{tab:results} for a comprehensive
summary of our results.

\begin{table}
\centering
\begin{tabular}{l|l|l|l|l}
Problem & Model & Old & New & Comment \\
\hline
\hline
continuous (D) & PM & $O(n^2)$
& $O\left(n^2\, \frac{(\log\log n)^{3/2}}{\sqrt{\log n}}\right)$\\
\hline
continuous (D) & WRAM & $O(n^2)$
& $O\left(n^2\, \frac{(\log\log n)^{2}}{\log n}\right)$\\
\hline
continuous (D)  & DT & 
$O(n^2)$&
$O\left(n^{2 - \eps}\right)$&
for $0 < \eps < 1/6$\\
\hline
continuous (C)  & PM & $O(n^2\log n)$
& $O\left(n^2\sqrt{\log n}(\log\log n)^{3/2}\right)$\\
\hline
continuous (C)  & WRAM & $O(n^2\log n)$
& $O\left(n^2(\log\log n)^{2}\right)$\\
\hline
discrete (D)  & PM & 
$O\left(n^2\right)$
&$O\left(n^2 \, \frac{\log\log n}{\log n}\right)$
& \cite{AgarwalBAKaSh14} uses WRAM \\
\hline
discrete (C)  & DT & 
$O(n^2)$&
$O\left(n^{4/3}\log^c n\right)$ &
implicit in \cite{AgarwalBAKaSh14}\\
\hline
weak (D)  & PM & 
$O(n^2)$&
$O\left(n^2 \, \frac{\alpha(n)\log\log n}{\log n}\right)$\\
\hline
weak (D)  & WRAM & 
$O(n^2)$&
$O\left(n^2 \, \frac{(\log\log n)^5}{\log^2 n}\right)$\\
\hline
weak (C)  & PM & 
$O(n^2\log n)$&
$O\left(n^2 \alpha(n)\log\log n\right)$\\
\hline
weak (C)  & WRAM & 
$O(n^2\log n)$&
$O\left(n^2 \frac{(\log\log n)^5}{\log n}\right)$\\
\end{tabular}
\caption{Summary of results. We distinguish the
continuous and the discrete \frechet distance in the 
decision (D) and the computation (C) version.
We also consider the weak continuous \frechet distance.
The computational models are the pointer machine (PM),
the word RAM (WRAM) or the algebraic decision trees (DT).
The old bounds are due to Alt and Godau~\cite{AltGo95}
(continuous and weak) and Eiter and 
Mannila~\cite{EiterMannila94} (discrete).}
\label{tab:results}
\end{table}

\paragraph{Recent developments}
Recently, Ben Avraham~\etal~\cite{BenAvFiKaKaSh15}
presented a subquadratic algorithm for the discrete 
\frechet distance with shortcuts that runs in 
$O(n^{4/3}\log^3 n)$ time. This running time resembles,
at least superficially, our result on algebraic computation
trees for the general discrete \frechet distance.

When we initially announced our results, we believed 
that they provided strong evidence that Alt's conjecture
is false. Indeed, for a long time it was conjectured 
that no subquadratic decision tree exists for 
3SUM~\cite{Patrascu:2010} and an $\Omega (n^2)$ lower 
bound is known in a restricted linear decision tree 
model~\cite{AilonCh05,Erickson99}. However, in a
recent--and in our opinion quite astonishing--result,
Gr\o{}nlund and Pettie showed that if we allow only
slightly more powerful algebraic decision trees than
in the previous lower bounds, one can decide 3SUM 
non-uniformly in $O(n^{2-\eps})$ steps,
for some fixed $\eps > 0$~\cite{GronlundPe14}.
They also show that this leads to a general subquadratic
algorithm for 3SUM, a situation very similar to
the \frechet distance as described in the present
paper. Thus, despite some interesting developments,
the status of Alt's conjecture remains as open as
before. However, we can now see that there 
exists a wide variety of efficiently solvable problems 
such as (in addition to 3SUM and the \frechet distance) 
\textsc{Sorting $X+Y$}~\cite{Fredman75},
\textsc{Min-Plus-Convolution}~\cite{BremnerChDeErHuIaLaTa12}, 
or finding the Delaunay triangulation for a point set 
that has been sorted in two orthogonal
directions~\cite{BuchinMu11}, for which there seems to be
a noticeable gap between the decision tree complexity
and the uniform complexity.

In our initial announcement, we also asked whether, 
besides 3SUM-hardness, there may be other reasons 
to believe that the quadratic running time for the
\frechet distance cannot be improved. Karl Bringmann provided
an interesting answer to this question by showing that
any algorithm for the \frechet distance with running 
time $O(n^{2-\eps})$, for some fixed $\eps > 0$,
would violate the \emph{strong exponential time hypothesis}
(SETH)~\cite{Bringmann14}.
These results were later refined and improved to show that
the lower bound holds in basically all settings (with the 
notable exception of the 
one-dimensional continuous \frechet 
distance, which is still unresolved)~\cite{BringmannMu16}. 
We believe that these developments show that
the \frechet distance still holds many interesting
aspects to be discovered and remains an intriguing
object of further study. 

\begin{figure*}
\centering
\includegraphics{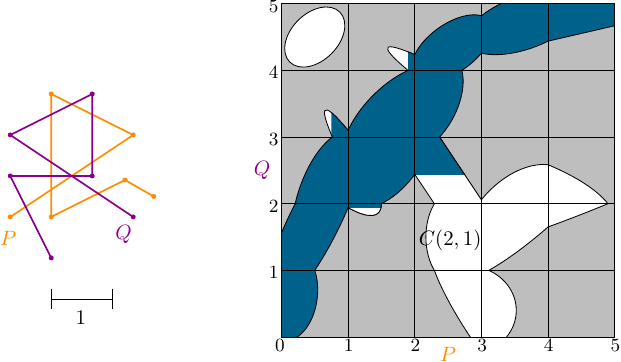}
\caption{Two polygonal curves $P$ and $Q$, together 
with their associated free-space diagram.
The reachable region $\reach(P,Q)$ is shown in blue.
For example, the white area in $C(2,1)$, denoted $F(2,1)$, 
corresponds to all points on the third edge of $P$ 
and the second edge of $Q$ that have distance at most $1$.
As $(5,5) \in \reach(P,Q)$, we have $d_F(P,Q) \leq 1$.}
\label{fig:freespace}
\end{figure*}

\section{Preliminaries and Basic Definitions}
\label{sec:prelim}

Let $P$ and $Q$ be two polygonal curves in the plane, 
defined by their vertices $p_0, p_1, \dots, p_n$ and 
$q_0, q_1, \dots, q_n$.  Depending on the context, 
we interpret $P$ and $Q$ either as sequences of $n$ 
and $n$ edges, or as continuous functions 
$P\colon [0, n] \rightarrow \R^2$ and 
$Q\colon [0, n] \rightarrow \R^2$.  In the latter 
case, we have $P(i + \lambda) = (1-\lambda)p_i + \lambda p_{i+1}$ 
for $i = 0, \dots, n-1$ and $\lambda \in [0,1]$, 
and similarly for $Q$.  Let $\Psi$ be the set of 
all continuous and nondecreasing functions 
$\alpha\colon [0,1] \rightarrow [0,n]$ with $\alpha(0) = 0$
and $\alpha(1) = n$. 
The \emph{\frechet distance}
between $P$ and $Q$ is defined as
\[
d_F(P,Q) \eqdef
\inf_{\alpha, \beta \in \Psi} \max_{x \in [0,1]} \| P(\alpha(x)) - Q(\beta(x)) \|,
\]
where $\|\cdot\|$ denotes the Euclidean distance.

The classic approach to computing $d_F(P, Q)$ uses 
the \emph{free-space diagram} $\FSD(P, Q)$. It is 
defined as
\[
  \FSD(P,Q) \eqdef \{(x,y) \in [0,n] \times [0,n] \mid \|P(x) - Q(y) \| \leq 1\}.
\]
In other words, $\FSD(P, Q)$ is the subset of the 
joint parameter space for $P$ and $Q$ where the 
corresponding points on the curves have distance
at most $1$, see Figure~\ref{fig:freespace}.

The structure of $\FSD(P,Q)$ is easy to describe.
Let $R \eqdef [0,n] \times [0, n]$ be the ground 
set. We subdivide $R$ into $n^2$ \emph{cells}
$C(i,j) = [i, i+1] \times [j, j+1]$, for 
$i,j = 0, \dots, n-1$.
The cell $C(i,j)$ corresponds to the edge pair 
$e_{i+1}$ and $f_{j+1}$, where $e_{i+1}$ is the 
$(i+1)$\textsuperscript{th} edge of $P$ and
$f_{j+1}$ is the $(j+1)$\textsuperscript{th} edge
of $Q$.  Then the set 
$F(i,j) \eqdef \FSD(P,Q) \cap C(i,j)$ represents
all pairs of points on $e_{i+1} \times f_{j+1}$
with distance at most $1$.  Elementary geometry
shows that $F(i,j)$ is the intersection of 
$C(i,j)$ with an ellipse~\cite{AltGo95}. In
particular, the set $F(i,j)$ is convex,
and the intersection of $\FSD(P,Q)$ with the
boundary of $C(i,j)$ consists of four (possibly
empty) intervals, one on each side of $\bd C(i,j)$.
We call these intervals the \emph{doors} of $C(i,j)$
in $\FSD(P,Q)$. A door is said to be \emph{closed}
if the interval is empty, and \emph{open} otherwise.

A path $\pi$ in $\FSD(P,Q)$ is \emph{bimonotone} if 
it is both $x$- and $y$-monotone, i.e., every vertical
and every horizontal line intersects $\pi$ in at
most one connected component. Alt and Godau observed
that it suffices to decide whether there exists a
bimonotone path from $(0,0)$ to $(n,n)$ inside
$\FSD(P,Q)$. We define the \emph{reachable region} 
$\reach(P,Q)$ as the set of points in $\FSD(P,Q)$ 
that are reachable from $(0,0)$ on a bimonotone path.
Then, $d_F(P,Q) \leq 1$ if and only if $(n,n) \in \reach(P,Q)$,
see Figure~\ref{fig:freespace}.  It is not necessary
to compute all of $\reach(P,Q)$: since $\FSD(P,Q)$
is convex inside each cell, we only need 
the intersections $\reach(P,Q) \cap \bd C(i,j)$.
The sets defined by $\reach(P,Q) \cap \bd C(i,j)$ are
subintervals of the doors of the free-space
diagram, and they are defined by endpoints of 
doors in the free-space diagram in the same row or column.
We call the intersection of a door with $\reach(P,Q)$ a 
\emph{reach-door}. The reach-doors can be found in 
$O(n^2)$ time through a simple breadth-first-traversal of the 
cells~\cite{AltGo95}.  In the next sections, we 
show how to obtain the crucial information, 
i.e., whether $(n,n) \in \reach(P,Q)$, in $o(n^2)$
time instead.

\paragraph{Basic approach and intuition}
In our algorithm for the decision problem, we basically 
want to compute $\reach(P,Q)$. But instead of propagating
the reachability information cell by cell, we always group 
$\tau \times \tau$ cells (with $1 \ll \tau \ll n$) 
into an \emph{elementary box} of cells. 
When processing a box, we can assume that we know which 
parts of the left and the bottom boundary of the box are
reachable. That is, we know the reach-doors on the bottom
and left boundary, and we need to compute the reach-doors
on the top and right boundary of the elementary box.
These reach-doors are determined by the combinatorial
structure of the box. More specifically, suppose we
know for every row and column the order of the door endpoints 
(including for the reach-doors on the left and bottom boundary).
Then, we can deduce which of these door boundaries 
determine the reach-doors on the top and right boundary. 
We call the sequence of these
orders, the \emph{(full) signature} of the box.

The total number of possible signatures is bounded by an
expression in terms of $\tau$. Thus, if we 
pick $\tau$ sufficiently small compared to
$n$, we can pre-compute for all possible signatures
the reach-doors on the top and right boundary, and build a
data structure to query these quickly (Section~\ref{sec:lookup}).
Since the reach-doors on the bottom and left boundary are
required to make the signature, we initially have only
incomplete signatures. In Section~\ref{sec:preproc2}, we
describe how to compute these efficiently. The incomplete
signatures are then used to preprocess the data structure such
that we can quickly find the full signature once we know
the reach-doors of an elementary box. After building and
preprocessing the data structure, it is possible to determine
$d_F(P,Q) \leq 1$ efficiently by traversing the free-space
diagram elementary box by elementary box, as explained
in Section~\ref{sec:procFSD}.

\section{Building a Lookup Table}
\label{sec:lookup}

\subsection{Preprocessing an elementary box}

Before it considers the input, our algorithm builds 
a lookup table. As mentioned above, the purpose of 
this table is to speed up the computation
of small parts of the free-space diagram.

\begin{figure}[t]
\centering
\includegraphics[scale=0.85]{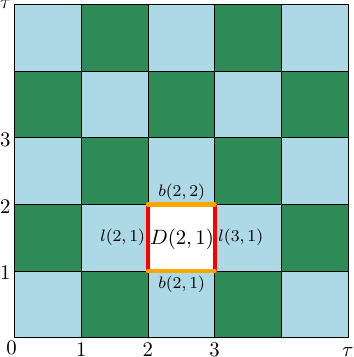}
\caption{The elementary box. The cell $D(2,1)$ is shown white.
Its boundaries---$l(2,1)$, $l(3,1)$, $b(2,1)$, $b(2,2)$---are indicated.}
\label{fig:elemgrid}
\end{figure}

\begin{figure*}[ht]
\centering
\includegraphics{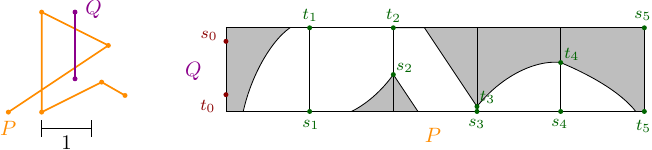}
\caption{The door-order of a row (the vertical
order of the points) encodes the combinatorial
structure of the doors. The door-order for the 
row in the figure is $s_1 s_3 s_4 t_5 t_3 t_0 s_2 t_4 s_0 s_5 t_1 t_2$. 
Note that $s_0$ and $t_0$ represent the reach-door, 
which is empty in this case. These are omitted
in the incomplete door-order.}
\label{fig:signature}
\end{figure*}
Let $\tau \in \N$ be a parameter.\footnote{A
preview for the impatient reader: we
later set $\tau = \Theta(\sqrt{\log n/\log\log n})$.}
The \emph{elementary box} is a subdivision
of $[0,\tau]^2$ into $\tau$ columns
and rows, thus $\tau^2$ cells.\footnote{For now,
the elementary box is a combinatorial concept. In
the next section, we overlay these boxes on
the free-space diagram to obtain ``concrete''
elementary boxes.}
For $i,j = 0, \dots, \tau-1$, we denote the
cell $[i,i+1] \times [j, j+1]$ with $D(i,j)$.  We
denote the left side of the boundary
$\bd D(i,j)$ by $l(i,j)$ and the bottom
side by $b(i,j)$.  Note that $l(i,j)$ coincides
with the right side of $\bd D(i-1, j)$ and $b(i,j)$
with the top of $\bd D(i, j-1)$.  Thus, we write
$l(\tau, j)$ for the \emph{right} side of $D(\tau - 1, j)$ and
$b(i, \tau)$ for the \emph{top} side of $D(i, \tau - 1)$.
Figure~\ref{fig:elemgrid} shows the elementary box.

The \emph{door-order} $\sigma_j^r$ for a row $j$ is a
permutation of $\{s_0, t_0, \dots, s_{\tau}, t_{\tau}\}$,
having $2\tau+2$ elements.  For $i = 1, \dots, \tau$,
the element $s_i$ represents the lower endpoint of
the door on $l(i,j)$, and $t_i$ represents the upper
endpoint.  The elements $s_0$ and $t_0$ are an
exception: they describe the reach-door on the
boundary $l(0, j)$ (i.e., its intersection with
$\reach(P,Q)$).  The door-order $\sigma_j^r$ represents
the combinatorial order of these endpoints, as
projected onto a vertical line, i.e.,
they are sorted into their vertical order.
Some door-orders may encode the same combinatorial structure.
In particular, when door $i$ is closed, the exact position
of $s_i$ and $t_i$ in a door-order is irrelevant, as long as $t_i$
comes before $s_i$.  For a closed door $i$ ($i > 0$), we
assign $s_i$ to the upper endpoint of $l(i,j)$ and $t_i$ to the
lower endpoint. The values of $s_0$ and $t_0$ are
defined by the reach-door and their relative order
is thus a result of computation.  We break ties
between $s_i$ and $t_{i'}$ by placing $s_i$ before $t_{i'}$, and any
other ties are resolved by index.
A door-order $\sigma_i^c$ is defined analogously for a column $i$.
We write $x <_i^c y$ if $x$ comes before $y$ in
$\sigma_i^c$, and $x <_j^r y$ if $x$ comes before $y$ in
$\sigma_j^r$. An \emph{incomplete door-order} is a door-order in 
which $s_0$ and $t_0$ are omitted (i.e. the intersection of 
$\reach(P,Q)$ with the door is
still unknown); see Figure~\ref{fig:signature}.

We can now define the \emph{(full) signature} of the elementary box as
the aggregation of
the door-orders of its rows and columns.  Therefore, a signature
$\Sigma = (\sigma_1^c, \dots, \sigma_\tau^c, \sigma_1^r, \dots,
\sigma_\tau^r)$ consists of $2\tau$ door-orders: one door-order
$\sigma_i^c$ for each column $i$ and one door-order $\sigma_j^r$ for
each row $j$ of the elementary box. Similarly, an 
\emph{incomplete signature} is the aggregation
of incomplete door-orders.

For a given signature, we define the \emph{combinatorial reachability structure} of the elementary box as follows.
For each column $i$ and for each row $j$, the combinatorial reachability structure indicates
which door boundaries in the respective column or row
define the reach-door of $b(i, \tau)$ or $l(\tau, j)$.

\begin{lemma}\label{lem:reachability}
Let $\Sigma$ be a signature for the elementary box.
Then we can determine the combinatorial reachability 
structure of the box in total time $O(\tau^2)$.
\end{lemma}
\begin{proof}
We use dynamic programming, very similar to the algorithm by Alt and
Godau~\cite{AltGo95}.
For each vertical edge $l(i,j)$ we define a variable $\widehat{l}(i,j)$, and
for each horizontal edge $b(i,j)$ we define a variable $\widehat{b}(i,j)$.
The $\widehat{l}(i,j)$ are pairs of the form $(s_u, t_v)$, representing the
reach-door $\reach(P,Q) \cap l(i,j)$.
If this reach-door is closed, then $t_v <_j^r s_u$ holds.
If the reach-door is open, then it is bounded by the lower endpoint of the
door on $l(u,j)$ and by the upper endpoint of the door on $l(v,j)$. (Note
that in this case we have $v = i$.)
Once again $s_0$ and $t_0$ are special and represent the reach-door on
$l(0,j)$.
The variables $\widehat{b}(i,j)$ are defined analogously.

Now we can compute $\widehat{l}(i,j)$ and $\widehat{b}(i,j)$ recursively
as follows: first, we set
\[
\widehat{l}(0,j) = \widehat{b}(i,0) = (s_0, t_0), \quad \textnormal{for }
i,j = 0, \dots, \tau-1.
\]
Next, we describe how to find $\widehat{l}(i,j)$ given $\widehat{l}(i-1, j)$
and $\widehat{b}(i-1, j)$, see Figure~\ref{fig:recursion}.

\begin{figure}[t]
\centering
\includegraphics{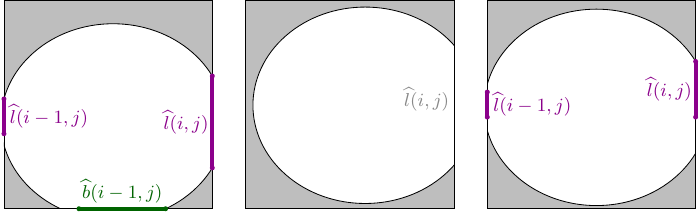}
\caption{The three cases for the recursive definition of 
$\widehat{l}(i,j)$.
If the lower boundary is reachable, we can reach the whole right
door (left). If neither the lower nor the left boundary is reachable,
the right door is not reachable either (middle). Otherwise,
the lower boundary is the maximum of
$\widehat{l}(i-1,j)$ and the lower boundary of the right door (right).}
\label{fig:recursion}
\end{figure}

\textbf{Case 1:} Suppose $\widehat{b}(i-1, j)$ is open. This means that
$b(i-1, j)$ intersects $\reach(P,Q)$, so $\reach(P,Q) \cap l(i,j)$ is
limited only by the door on $l(i,j)$, and we can set
$\widehat{l}(i,j) \eqdef (s_i, t_i)$.

\textbf{Case 2:} If both $\widehat{b}(i-1, j)$ and 
$\widehat{l}(i-1, j)$ are
closed, it is impossible to reach $l(i,j)$ and thus we set
$\widehat{l}(i,j) \eqdef \widehat{l}(i-1,j)$.

\textbf{Case 3:}
If $\widehat{b}(i-1, j)$ is closed and $\widehat{l}(i-1, j)$ is open, we
may be able to reach $l(i,j)$ via $l(i-1,j)$.
Let $s_u$ be the lower endpoint of $\widehat{l}(i-1, j)$. We need to pass
$l(i,j)$ above $s_u$ and $s_i$ and below $t_i$, and therefore set
$\widehat{l}(i,j) \eqdef (\max(s_u, s_i), t_i)$, where the maximum
is taken according to the order $<_j^r$.

The recursion for the variable $\widehat{b}(i,j)$ is defined similarly.
We can implement the recursion in
time $O(\tau^2)$ for any given signature, for example by traversing the
elementary box column by column, while processing each column from bottom
to top.
\end{proof}

There are at most $((2\tau+2)!)^{2\tau} = \tau^{O(\tau^2)}$ 
distinct signatures for the elementary box.
We choose $\tau = \lambda \sqrt{\log n/\log\log n}$ 
for a sufficiently small constant $\lambda>0$, so 
that this number becomes $o(n)$. Thus, during the 
preprocessing stage we have time to enumerate
all possible signatures and determine the 
corresponding combinatorial reachability structure 
inside the elementary box. This information is then 
stored in an appropriate data structure.

\subsection{Building the data structure}\label{ssec:building}

Before we describe this data structure, we first 
explain how the door-orders are represented. This 
depends on the computational model. By our choice 
of $\tau$, there are $o(n)$ distinct door-orders.
On the word RAM, we represent each door-order and 
incomplete door-order by an integer between $1$ and 
$(2\tau)!$. This fits into a word of $\log n$ bits.
On the pointer machine, we create a record for each 
door-order and incomplete door-order; we represent 
an order by a pointer to the corresponding record.

The data structure has two \emph{stages}.
In the first stage, 
we assume we know the incomplete door-order 
for each row and for each column of
the elementary box\footnote{In the next 
section, we describe how to determine the
incomplete door-orders efficiently.}, and 
we wish to determine the incomplete signature.
In the second stage  
we have obtained the reach-doors for the left and bottom
sides of the elementary box, and we are looking for the 
full signature.
The details of our method depend on the computational model.
One way uses table lookup and requires the word RAM; the other
way works on the pointer machine, but is a bit more involved.

\paragraph{Word RAM}
  We organize the lookup table as a large tree $T$.
  In the first stage,
  each level of $T$ corresponds to a row or column of the elementary
  box. Thus, there are $2\tau$ levels.
  Each node has $(2\tau)!$ children, representing the possible incomplete
  door-orders for the next row or column.
  Since we represent door-orders by positive integers, each node of $T$
  may store an array for its children; we
  can choose the appropriate child for a given incomplete door-order
  in constant time. Thus, determining the incomplete signature for an
  elementary box requires $O(\tau)$ steps on a word RAM.

  For the second stage, we again use a tree structure. Now
  the tree has $O(\tau)$ \emph{layers}, each with $O(\log \tau)$ levels.
  Again, each layer corresponds to a row or column of the elementary box.
  The levels inside each layer then implement a balanced binary search tree
  that allows us to locate the endpoints of the reach-door within the
  incomplete signature. Since there are $2\tau$ endpoints, this
  requires $O(\log \tau)$ levels.
  Thus, it takes $O(\tau \log \tau)$ time to find the full signature of a
  given elementary box.

\paragraph{Pointer machine}
Unlike in the word RAM model, we are not allowed to store
a lookup table on every level of the tree $T$, and there is no way to
quickly find the appropriate child for a given door-order. Instead,
we must rely on batch processing to achieve a reasonable
running time.

Thus, suppose that during the first stage we want to
find the incomplete signatures for a set $B$ of  $m$
elementary boxes, where again for each box in $B$ we know the
incomplete door-order for each row and
each column. Recall that we represent the door-order by a pointer to the
corresponding record. With each such record, we store a
queue of elementary boxes that is empty initially.

We now simultaneously propagate the boxes in $B$ through $T$,
proceeding level by level. In the first level, all of $B$
is assigned to the root of $T$.
Then, we go through the nodes of one level of $T$, from left to right.
Let $v$ be the current node of $T$.
We consider each elementary box $b$ assigned to $v$. We determine
the next incomplete door-order for $b$, and we append $b$
to the queue for this incomplete door-order---the queue is addressed
through the corresponding record, so all elementary boxes
with the same next incomplete door-order  end up in the same queue.
Next, we go through the nodes of the next level, again from left to right.
Let $v'$ be the
current node. The node $v'$ corresponds to a next incomplete  door-order
$\sigma$ that extends the known signature of its parents. We consider
the queue stored at the record for $\sigma$. By construction, the
elementary boxes that should be assigned to $v'$  appear
consecutively at the beginning of this queue. We remove these boxes from
the queue and assign them to $v'$. After this, all the queues are empty,
and we can continue by propagating the boxes to the next level.
During this procedure, we traverse each node of $T$ a constant number of
times, and in each level of the $T$ we consider all the boxes in
$B$. Since $T$ has $o(n)$ nodes, the total running time is
$O(n + m\tau)$.

For the second stage, the data structure works just as in the 
word RAM case, because no table lookup is necessary.
Again, we need $O(\tau \log \tau)$ steps to process one box.
After the second stage, we obtain the combinatorial
reachability structure of the box in constant time since we 
precomputed this information for each box (Lemma~\ref{lem:reachability}).
Thus, we have shown the following lemma, independently of the 
computational model.

\begin{lemma}\label{lem:data-struct}
For $\tau = \lambda \sqrt{\log n/\log\log n}$, with a 
sufficiently small constant $\lambda>0$, we can 
construct in $o(n)$ time a data structure of size 
$o(n)$ such that
\begin{itemize}
\item given a set of $m$ elementary boxes 
where the incomplete door-orders are known, 
we can find the incomplete signature of each 
box in total time $O(n + m \tau)$;
\item given the incomplete signature and the 
reach-doors on the bottom and left boundary 
of an elementary box, we can find the full 
signature in $O(\tau \log \tau)$ time;
\item given the full signature of an elementary 
box, we can find the combinatorial
reachability structure of the box in constant 
time.
\end{itemize}
\end{lemma}

\section{Preprocessing a Given Input}
\label{sec:preproc2}

Next, we perform a second preprocessing phase 
that considers the input curves $P$ and $Q$.
Our eventual goal is to compute the intersection 
of $\reach(P,Q)$ with the cell boundaries, 
taking advantage of the data structure
from Section~\ref{sec:lookup}. For this, we 
aggregate the cells of $\FSD(P,Q)$ into (concrete) 
elementary boxes consisting of $\tau \times \tau$ cells.
There are $n^2/\tau^2$ such boxes.
We may avoid rounding issues by either 
duplicating vertices or handling a small part 
of $\FSD(P,Q)$ without lookup tables.

The goal is to determine the signature for 
each elementary box $S$. At this point, this 
is not quite possible yet, since the signature 
depends on the intersection of $\reach(P,Q)$ 
with the lower and left boundary of $S$. 
Nonetheless, we can find the incomplete signature, 
in which the positions of $s_0, t_0$ (the 
reach-door) in the (incomplete) door-orders 
$\sigma_i^r$, $\sigma_j^c$ are still to 
be determined.

\begin{figure*}
\centering
\includegraphics[scale=0.85]{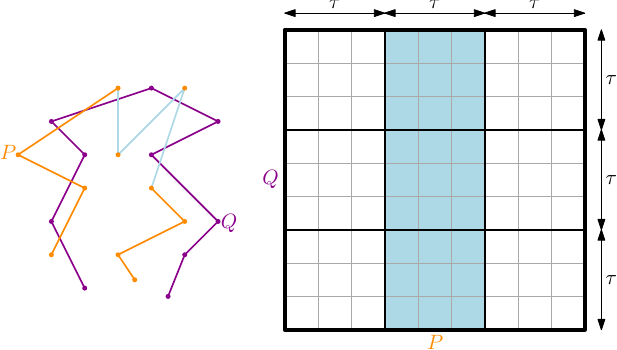}
\caption{$\FSD(P,Q)$ is subdivided into 
$n^2/\tau^2$ elementary boxes of size 
$\tau \times \tau$. The free-space diagram 
is subdivided into $n^2/\tau^2$ elementary 
boxes of size $\tau \times \tau$. A strip 
is a column of elementary boxes: it corresponds 
to a subcurve of $P$ with $\tau$ edges.}
\label{fig:fsdstrips}
\end{figure*}
We aggregate the columns of $\FSD(P,Q)$ 
into vertical \emph{strips}, each 
corresponding to a single column of 
elementary boxes (i.e., $\tau$ 
consecutive columns of cells in 
$\FSD(P,Q)$). See Figure~\ref{fig:fsdstrips}.

Let $A$ be such a strip. It corresponds 
to a subcurve $P'$ of $P$ with $\tau$ 
edges. The following lemma implies that 
we can build a data structure for $A$ 
such that, given any segment of $Q$, we 
can efficiently find its incomplete door-order 
within the elementary box in $A$.

\begin{lemma}\label{lem:preprocess_input}
Given a subcurve $P'$ with $\tau$ edges, 
we can compute in $O(\tau^6)$ time a data 
structure that requires $O(\tau^6)$ space 
and that allows us to determine the 
incomplete door-order of any line 
segment on $Q$ in time $O(\log \tau)$.
\end{lemma}

\begin{proof}
Consider the arrangement $\mathcal{A}$ 
of unit circles whose centers are the 
vertices of $P'$ (see Figure~\ref{fig:arranagement}).
The incomplete door-order of a line segment 
$s$ is determined by the intersections 
of $s$ with the arcs of $\mathcal{A}$ 
(and for a circle not intersecting $s$ 
by whether $s$ lies inside or outside 
of the circle).  Let $\ell_s$ be the 
line spanned by line segment $s$.
Suppose we wiggle $\ell_s$. The order 
of intersections of $\ell_{s}$ and the 
arcs of $\mathcal{A}$ changes only when 
$\ell_s$ moves over a vertex of $\mathcal{A}$ 
or if $\ell_s$ leaves or enters a circle.

\begin{figure*}[t]
\centering
\includegraphics{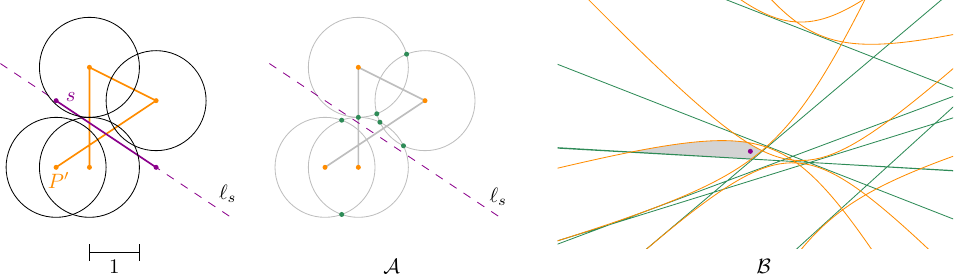}
\caption{By using the arrangement 
$\mathcal{A}$ defined by unit circles
centered at vertices of $P'$, we can 
determine the incomplete door-order of each
segment $s$ on $Q$. This is done by locating 
the dual point of $\ell_s$ in the dual 
arrangement $\mathcal{B}$. The dual 
arrangement also contains pseudolines to 
determine when $\ell_s$ leaves a circle 
of $\mathcal{A}$.}
\label{fig:arranagement}
\end{figure*}

We use the standard duality transform 
that maps a line $\ell: y = ax + b$ to 
the point $\ell^* : (a, -b)$, and vice 
versa.  Consider a unit circle $C$ in 
$\mathcal{A}$ with center $(c_x, c_y)$.
Elementary geometry shows that the set 
of all lines that are tangent to $C$ 
from above dualizes to the curve 
$t_a^*(C): y = c_x x - c_y - \sqrt{1+x^2}$. 
Similarly, the lines that are tangent to
$C$ from below dualize to the curve
$t_b^*(C): y = c_x x - c_y + \sqrt{1+x^2}$.
Define $C^* \eqdef \{t_a^*(C), t_b^*(C) \mid C \in \mathcal{A}\}$.
Since any pair of distinct circles
$C_1$, $C_2$ has at most four common 
tangents, one for each choice of 
above/below $C_1$ and above/below $C_2$, 
it follows that any two curves in $C^*$
intersect at most once.

Let $V$ be the set of vertices in $\mathcal{A}$, 
and let $V^*$ be the lines dual to the points 
in $V$ (note that $|V| = O(\tau^2)$).
Since for any vertex $v \in V$ and any circle
$C \in \mathcal{A}$ there are at most two tangents 
through $v$ on $C$, each line in $V^*$ intersects 
each curve in $C^*$ at most once. Thus, the 
arrangement $\mathcal{B}$ of the curves in
$V^* \cup C^*$ is an arrangement of \emph{pseudolines} 
with complexity $O(\tau^4)$. Furthermore, it 
can be constructed in the same expected time, 
together with a point location structure
that finds the containing cell in $\mathcal{B}$ 
of any given point in time 
$O(\log \tau)$~\cite[Chapter~6.6.1]{SharirAg95}.

Now consider a line segment $s$ and the supporting 
line $\ell_s$.  As observed in the first paragraph, 
the combinatorial structure of the intersection 
between $\ell_s$ and $\mathcal{A}$ is completely 
determined by the cell of $\mathcal{B}$ that
contains the dual point $\ell_s^*$.
Thus, for every cell $f(s) \in \mathcal{B}$,
we construct a list $L_{f(s)}$ that represents 
the combinatorial structure of $\ell_s \cap \mathcal{A}$.
There are $O(\tau^4)$ such lists, each having size
$O(\tau)$. We can compute $L_{f(s)}$ by traversing 
the zone of $\ell_s$ in $\mathcal{A}$.
Since circles intersect at most twice and since 
a line intersects any circle at most twice, 
the zone has complexity $O(\tau 2^{\alpha (\tau)})$, 
where $\alpha(\cdot)$ denotes the inverse Ackermann
function~\cite[Theorem 5.11]{SharirAg95}. Since
$O(\tau 2^{\alpha (\tau)}) \subset O(\tau^2)$, 
we can compute all lists in $O(\tau^{6})$ time.

Given the list $L_{f(s)}$, the incomplete door-order 
of $s$ is determined by the position of the 
endpoints of $s$ in $L_{f(s)}$. There are
$O(\tau^2)$ possible ways for this, and we 
build a table $T_{f(s)}$ that represents them. 
For each entry in $T_{f(s)}$, we store a 
representative for the corresponding incomplete 
door-order. As described in the previous 
section, the representative is a positive integer 
in the word RAM model and a pointer to the 
appropriate record on a pointer machine.

The total size of the data structure is 
$O(\tau^{6})$ and it can be constructed 
in the same time. A query works as follows:
given $s$, we can compute $\ell_s^*$ in
constant time. Then we use the point location 
structure of $\mathcal{B}$ to find $f(s)$ in 
$O(\log \tau)$ time. Using binary search on 
$T_{f(s)}$ (or an appropriate tree structure 
in the case of a point machine), we can then 
determine the position of the endpoints of $s$ 
in the list $L_{f(s)}$ in $O(\log \tau)$ time.
This bound holds both on the word RAM and
on the pointer machine.
\end{proof}

\begin{lemma}\label{lem:incomplete}
Given the data structure of Lemma~\ref{lem:data-struct}, 
the incomplete signature for each elementary 
box can be determined in time
$O(n\tau^{5} + n^2 (\log \tau) / \tau)$.
\end{lemma}

\begin{proof}
By building and using the data structure 
from Lemma~\ref{lem:preprocess_input}, we 
determine the incomplete door-order for each 
row in each vertical $\tau$-strip
in total time proportional to
\[
  \frac{n}{\tau}(\tau^6 + n \log \tau) = 
    n \tau^{5} + \frac{n^2 \log\tau}{\tau}.
\]

We repeat the procedure with the horizontal strips. 
Now we know for each elementary box in $\FSD(P,Q)$ 
the incomplete door-order for each row and each column. 
We use the data structure of Lemma~\ref{lem:data-struct} 
to combine these.  As there are $n^2 / \tau^2$ 
boxes, the number of steps is $O(n^2/\tau + n) = 
O(n^2/\tau)$.  Hence, the incomplete signature for each 
elementary box is found in $O(n\tau^{5} + n^2 (\log \tau) / \tau)$
steps.
\end{proof}

\section{Solving the Decision Problem}
\label{sec:procFSD}

With the data structures and preprocessing from the 
previous sections, we have all ingredients in place 
to determine whether $d_F(P,Q) \leq 1$. We know for 
each elementary box its incomplete signature and we have 
a data structure to derive its full signature (and with 
it, the combinatorial reachability structure) when its 
reach-doors are known. What remains to be shown is that 
we can efficiently process the free-space diagram to 
determine whether $(n,n) \in \reach(P,Q)$. This is captured 
in the following lemma. 

\begin{lemma}\label{lem:procFSD}
If the incomplete signature for each elementary box is 
known, we can determine whether $(n,n) \in \reach(P,Q)$ 
in time $O(n^2(\log \tau)/\tau)$.
\end{lemma}

\begin{proof}
We go through all elementary boxes of $\FSD(P,Q)$, 
processing them one column at a time, going from bottom 
to top in each column. Initially, we know the full 
signature for the box $S$ in the lower left corner 
of $\FSD(P,Q)$. We use the signature to determine the 
intersections of $\reach(P,Q)$ with the upper and right 
boundary of $S$. There is a subtlety here: the 
signature gives us only the combinatorial reachability 
structure, and we need to map the resulting $s_i, t_j$ 
back to the corresponding vertices on the curves. On
the word RAM, this can be done easily through table lookups. 
On the pointer machine, we use representative records for 
the $s_i, t_i$ elements and use $O(\tau)$ time before
processing the box to store a pointer from each representative 
record to the appropriate vertices on $P$ and $Q$.

We proceed similarly for the other boxes. By the choice 
of the processing order of the elementary boxes we always 
know the incoming reach-doors on the bottom and left 
boundary when processing a box.
Given the incoming reach-doors, we can determine the full 
signature and find the structure of the outgoing reach-doors 
in total time $O(\tau \log \tau)$, using Lemma~\ref{lem:data-struct}. 
Again, we need $O(\tau)$ additional time on the pointer machine 
to establish the mapping from the abstract
$s_i$, $t_i$ elements to the concrete vertices of $P$ and $Q$.
In total, we spend $O(\tau \log \tau)$ time per box. Thus, it takes time
$O(n^2(\log \tau)/\tau)$ to process all boxes, as claimed.
\end{proof}

As a result, we obtain the following theorem 
for the pointer machine (and, by extension, 
for the real RAM model). For the word RAM model, 
we can obtain an even faster algorithm (see 
Section~\ref{sec:wordram}).

\begin{theorem}\label{thm:decpointer}
There is an algorithm that solves the decision 
version of the \frechet problem in 
$O(n^2 (\log\log n)^{3/2}/\sqrt{\log n})$ time 
on a pointer machine.
\end{theorem}
\begin{proof}
Set $\tau = \lambda\sqrt{\log n/\log\log n}$, for 
a sufficiently small constant $\lambda>0$. 
The theorem follows by applying Lemmas~\ref{lem:data-struct},~\ref{lem:incomplete}, 
and~\ref{lem:procFSD} in sequence.
\end{proof}

\section{Improved Bound on the Word RAM}\label{sec:wordram}

We now explain how the running time of our algorithm 
can be improved if our computational model allows for
constant time table-lookup. We use the same $\tau$ as 
above (up to a constant factor). However, we change a 
number of things. ``Signatures'' are represented 
differently and the data structure to obtain combinatorial 
reachability structures is changed accordingly.
Furthermore, we aggregate elementary boxes into 
\emph{clusters} and determine ``incomplete door-orders'' 
for multiple boxes at the same time.  Finally, we walk 
the free-space diagram based on the clusters to decide 
$d_F(P,Q) \leq 1$.

\begin{figure}[t]
\centering
\includegraphics{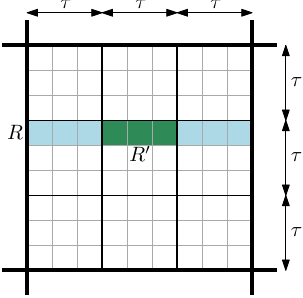}
\caption{A cluster consists of $\tau \times \tau$ 
elementary boxes, thus of $\tau^2 \times \tau^2$ cells. 
A row $R$ and its corresponding $R'$ for the central 
elementary box are indicated.}
\label{fig:cluster}
\end{figure}

\paragraph{Clusters and extended signatures}
We introduce a second level of aggregation in 
the free-space diagram (see Figure~\ref{fig:cluster}): 
a \emph{cluster} is a collection of $\tau \times \tau$ 
elementary boxes, that is, $\tau^2 \times \tau^2$ 
cells in $\FSD(P,Q)$. Let $R$ be a row of cells 
in $\FSD(P,Q)$ of a certain cluster. As before, 
the row $R$ corresponds to an edge $e$ on $Q$ 
and a subcurve $P'$ of $P$ with $\tau^2$ edges.
We associate with $R$ an ordered set 
$Z = \langle e_0, z_0', z_1,  z_1', z_2, z_2', 
\dots, z_{k}, z_{k}', e_1 \rangle$ with 
$2 \cdot k + 3$ elements. Here $k$ is the 
number of intersections of $e$ with the unit 
circles centered at the $\tau^2$ vertices of 
$P'$ (all but the very first). Hence, $k$ is 
bounded by $2 \tau^2$ and $|Z|$ is bounded by 
$4 \tau^2 + 3$. The order of $Z$ indicates the 
order of these intersections with $e$ directed 
along $Q$. Elements $e_0$ and $e_1$ represent 
the endpoints of $e$ and take a special role.
In particular, these are used to represent 
closed doors and snap open doors to the edge 
$e$. The elements $z_i'$ are placeholders for 
the positions of the endpoints of the reach-doors: 
$z_0'$ represents a possible reach-door endpoint 
between $e_0$ and $z_1$, the element $z_1'$ 
represents an endpoint between $z_1$ and $z_2$, etc.

Consider a row $R'$ of an elementary box inside 
the row $R$ of a cluster, corresponding to an 
edge $e$ of $Q$. The \emph{door-index} of $R'$ 
is an ordered set $\langle s_0, t_0, \ldots, 
s_\tau, t_\tau \rangle$ of size $2 \tau + 2$.  
Similar to a door-order, elements $s_0$ and $t_0$ 
represent the reach-door at the leftmost 
boundary of $R'$; the elements $s_i$ and 
$t_i$ ($1 \leq i \leq \tau$) represent the 
door at the right boundary of the $i$\textsuperscript{th} 
cell in $R'$. However, instead of rearranging 
the set to indicate relative positions, the 
elements $s_i$ and $t_i$ simply refer to 
elements in $Z$. If the door is open, they refer 
to the corresponding intersections with $e$ 
(possibly snapped to $e_0$ or $e_1$).  
If the door is closed, $s_i$ is set to $e_1$ 
and $t_i$ is set to $e_0$. The elements $s_0$ and 
$t_0$ are special, representing 
the reach-door, and they refer to one of 
the elements $z_i'$. An \emph{incomplete door-index} 
is a door-index without $s_0$ and $t_0$. 
The advantage of a door-index over a door-order 
is that the reach-door is always 
at the start.  Hence, completing an incomplete 
door-index to a full door-index can be done in 
constant time. Since a door-index has size 
$2 \tau + 2$, the number of possible door-indices 
for $R'$ is $\tau^{O(\tau)}$.

We define the door-indices for the columns 
analogously. We concatenate the door-indices 
for the rows and the columns to obtain the 
\emph{indexed signature} for an elementary 
box. Similarly, we define the \emph{incomplete 
indexed signature}.  The total number of 
possible indexed signatures remains 
$\tau^{O(\tau^2)}$.

For each possible incomplete indexed signature 
$\Sigma$ we build a lookup table $T_\Sigma$ 
as follows: the input is a word with 
$4\tau$ fields of $O(\log \tau)$ bits each.
Each field stores the positions in $Z$ of 
the endpoints of the ingoing reach-doors 
for the elementary box: $2\tau$ fields for 
the left side, $2\tau$ fields for the lower 
side. The output consists of a word that 
represents the indices for the elements in 
$Z$ that represent the outgoing reach-doors 
for the upper and right boundary of the box.
Thus, the input of $T_\Sigma$ is a word of 
$O(\tau \log \tau)$ bits, and $T_\Sigma$ 
has size $\tau^{O(\tau)}$. Hence, for all 
incomplete indexed signatures combined, the 
size is $\tau^{O(\tau^2)} = o(n)$ by our 
choice of $\tau$.

\paragraph{Preprocessing a given input}
During the preprocessing for a given input 
$P,Q$, we use \emph{superstrips} consisting 
of $\tau$ strips. That is, a superstrip 
is a column of clusters and consists of 
$\tau^2$ columns of the free-space diagram.
Lemma~\ref{lem:preprocess_input} still 
holds, albeit with a larger constant $c$ 
in place of $6$. The data structure gets 
as input a query edge $e$, and it returns 
in $O(\log \tau)$ time a word that contains 
$\tau$ fields. Each field represents the 
incomplete door-index for $e$ in the 
corresponding elementary box and thus 
consists of $O(\tau \log \tau)$ bits. Hence, 
the word size is $O(\tau^2 \log \tau) = O(\log n)$ 
by our choice of $\tau$. Thus, the total 
time for building a data structure for 
each superstrip and for processing all rows is
$O(n/\tau^2\;(\tau^c + n \log \tau)) = O(n^2(\log \tau)/\tau^2)$.
We now have parts of the incomplete indexed signature 
for each elementary box packed into different words.
To obtain the incomplete indexed signature, 
we need to rearrange the information such 
that the incomplete door-indices of the rows 
in one elementary box are in a single word.
This corresponds to computing a transpose 
of a matrix, as is illustrated in 
Figure~\ref{fig:transposition}.  For this, we 
need the following lemma, which can be 
found---in slightly different form---in Thorup~\cite[Lemma~9]{Thorup02}.

\begin{figure*}
\centering
\includegraphics{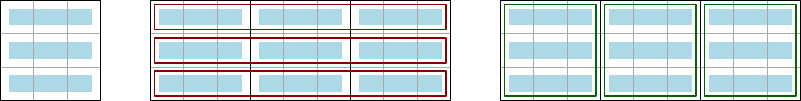}
\caption{(left) Every field represents 
the incomplete door-index of a row in an 
elementary box. (center) The fields 
are grouped into words per row in a 
cluster. (right) Transposition yields 
the desired organization, where a word 
represents the incomplete door-index of 
the rows in an elementary box.}
\label{fig:transposition}
\end{figure*}

\begin{lemma}\label{lem:transpose}
Let $X$ be a sequence of $\tau$ words 
that contain $\tau$ fields each, so 
that $X$ can be interpreted as a 
$\tau \times \tau$ matrix. Then we 
can compute in time $O(\tau \log \tau)$ 
on a word RAM a sequence $Y$ of $\tau$ 
words with $\tau$ fields each that 
represents the \emph{transpose} of $X$.
\end{lemma}

\begin{proof}
The algorithm is recursive and solves 
a more general problem: let $X$ be a 
sequence of $a$ words that represents a
sequence $M$ of $b$ different 
$a\times a$ matrices, such that the
$i$\textsuperscript{th} word in $X$ 
contains the fields of the $i$\textsuperscript{th} 
row of each matrix in $M$ from left to right.  
Compute a sequence of words $Y$ that 
represents the sequence $M'$ of the 
transposed matrices in $M$.

The recursion works as follows: if $a=1$, 
there is nothing to be done. Otherwise, 
we split $X$ into the sequence $X_1$ of 
the first $a/2$ words and the sequence 
$X_2$ of the remaining words. $X_1$ and $X_2$ 
now represent a sequence of $2b$ $(a/2) \times (a/2)$ 
matrices, which we transpose recursively. 
After the recursion, we put the $(a/2) \times (a/2)$ 
submatrices back together in the obvious way. 
To finish, we need to transpose the
off-diagonal submatrices. This can be done 
simultaneously for all matrices in time $O(a)$, 
by using appropriate bit-operations (or 
table lookup). Hence, the running time 
obeys a recursion of the form $T(a) = 2T(a/2) + O(a)$,
giving $T(a) = O(a \log a)$, as desired.
\end{proof}

By applying the lemma to the words that 
represent $\tau$ consecutive rows in
a superstrip, we obtain the incomplete door-indices 
of the rows for each elementary box.
This takes total time proportional to
\[
  \frac{n}{\tau^2} \cdot \frac{n}{ \tau} \cdot \tau \log \tau = 
  \frac{n^2}{\tau^{2}} \log \tau.
\]
We repeat this procedure for the horizontal 
superstrips. By using an appropriate lookup table
to combine the incomplete door-indices of the 
rows and columns, we obtain the incomplete 
indexed signature for each elementary box 
in total time $O(n^2(\log \tau)/\tau^2)$.

\paragraph{The actual computation}
We traverse the free-space diagram cluster 
by cluster (recall that a cluster consists 
of $\tau \times \tau$ elementary boxes).
The clusters are processed column by column 
from left to right, and inside each column 
from bottom to top. Before processing a 
cluster, we walk along the left and lower 
boundary of the cluster to determine the 
incoming reach-doors. This is done by 
performing a binary search for each box 
on the boundary, and determining the 
appropriate elements $z_i'$  which 
correspond to the incoming reach-doors.
Using this information, we assemble the 
appropriate words that represent the 
incoming information for each elementary 
box. Since there are $n^2/\tau^4$ clusters, 
this step requires time 
$O((n^2/\tau^4)\tau^2\log \tau) = O(n^2 (\log \tau) / \tau^2)$.
We then process the elementary boxes inside 
the cluster, in a similar fashion.
Now, however, we can process each elementary 
box in constant time through a
single table lookup, so the total time is $O(n^2/\tau^2)$.
Hence, the total running time of our algorithm is 
$O(n^2(\log \tau)/\tau^2)$.
By our choice of 
$\tau = \lambda\sqrt{\log n/ \log\log n}$ for 
a sufficiently
small $\lambda>0$, we obtain the following theorem.

\begin{theorem}\label{thm:decram}
The decision version of the \frechet 
problem can be solved in $O(n^2 (\log\log n)^{2}/\log n)$ 
time on a word RAM.
\end{theorem}

\section{Computing the \frechet Distance}
\label{sec:compute}

The optimization version of the \frechet problem,
i.e., computing the \frechet distance, can be
done in $O(n^2 \log n)$ time using parametric search
with the decision version as a subroutine~\cite{AltGo95}.
We showed that the decision problem can be solved in
$o(n^2)$ time. However, this does not directly yield
a faster algorithm for the optimization problem: if
the running time of the decision problem is $T(n)$,
parametric search gives an $O((T(n)+n^2) \log n)$
time algorithm~\cite{AltGo95}. There is an alternative
randomized algorithm by Raichel and
Har-Peled~\cite{HarPeledRa14}. Their algorithm also
needs $O((T(n)+n^2) \log n)$ time, but
below we adapt it to obtain the following lemma.

\begin{lemma}\label{lem:opt}
The \frechet distance of two polygonal curves
with $n$ vertices each can be computed by a
randomized algorithm in
$O(n^2 2^{\alpha(n)} + T(n) \log n)$ expected
time, where $T(n)$ is the time for the
decision problem.
\end{lemma}

Before we prove the lemma, we recall that
possible values of the \frechet distance
are limited to a certain set of
\emph{critical values}~\cite{AltGo95}:
\begin{enumerate}\setlength{\itemsep}{-3pt}
\item the distance between a vertex of one
curve and a vertex of the other curve
(\textbf{vertex-vertex});
\item the distance between a vertex of one
curve and an edge of the other curve
(\textbf{vertex-edge}); and
\item for two vertices of one curve and an
edge of the other curve, the distance
between one of the vertices and the
intersection of $e$ with the bisector of the
two vertices (if this intersection exists)
(\textbf{vertex-vertex-edge}).
\end{enumerate}

If we also include vertex-vertex-edge
tuples with no intersection, we can 
sample a critical value uniformly at 
random in constant time. The algorithm 
now works as follows (see Har-Peled 
and Raichel~\cite{HarPeledRa14} for 
more details): first, we sample a 
set $S$ of $K=4n^2$ critical values 
uniformly at random. Next, we find 
$a', b' \in S$ such that the \frechet 
distance lies between $a'$ and $b'$ and 
such that $[a', b']$ contains no other
value from $S$. In the original 
algorithm this is done by sorting $S$ 
and performing a binary search using 
the decision version. Using 
median-finding instead, this step 
can be done in $O(K + T(n) \log K)$ 
time. Alternatively, the running time 
of this step could be reduced by 
picking a smaller $K$. However, this 
does not improve the final bound, 
since it is dominated by a $O(n^2 2^{\alpha(n)})$ 
term. The interval $[a',b']$ with high 
probability contains only a small 
number of the remaining critical values.
More precisely, for $K=4n^2$ the 
probability that $[a', b']$ has more 
than $2 c n \ln n$ critical values is
at most $1/n^c$~\cite[Lemma 6.2]{HarPeledRa14}.

The remainder of the algorithm proceeds as
follows: first, we find all critical values of
type vertex-vertex and vertex-edge that lie
inside the interval $[a', b']$. This can
be done in $O(n^2)$ time by checking all 
vertex-vertex and vertex-edge pairs.
Among these values, we again use median-finding
to determine the interval $[a, b] \subseteq [a',b']$ 
that contains
the \frechet distance in $O(K' + T(n) \log K')$
time.  It remains to
determine the critical values corresponding 
to vertex-vertex-edge tuples that lie
in $[a, b]$. 

For this, take an edge $e$ 
of $P$ and the vertices of $Q$. 
Conceptually, we
start with circles of radius
$a$ around the vertices of $Q$, and we 
increase the radii until $b$. 
During this process, we observe
the evolution of the intersection 
points between the circle arcs and $e$. 
Because all vertex-vertex and 
vertex-edge events have been eliminated,
each circle intersects $e$ in either $0$
or $2$ points, and this does not change
throughout the process.
A critical value of 
vertex-vertex-edge type corresponds to the 
event that two different circles intersect 
$e$ in the same point, i.e., that two
intersection points meet while 
growing the circles.
Two intersection points can meet at most
once, and when they do, they exchange their
order along $e$.

This suggests the following algorithm:
let $\mathcal{A}_a$ be the arrangement of
circles with radius $a$ around the vertices
of $Q$, and let $\mathcal{A}_b$ be the concentric
arrangement of circles with radius $b$. 
We determine the 
ordered sequence $I_a$ of the
intersection points of the circles 
in $\mathcal{A}_a$ with $e$,
and we number them in their order along $e$.
Next, we find the ordered sequence 
of intersection points $I_b$ between
$e$ and the circles in  $\mathcal{A}_b$.
We assign to each point in $I_b$ the
number of the corresponding intersection 
points in $I_a$. Since $|I_a| = |I_b|$, 
this gives a permutation  of 
$\{1, \dots, |I_a|\}$,
Two intersection points change
their order from $I_a$ to $I_b$ 
exactly if there is a 
vertex-vertex-edge event in $[a,b]$,
so these events correspond to the inversions
of the resulting permutation.
Given that there are $k$ such inversions,
we can find them in
time $O(|I_a| + k)$ using insertion sort.
Thus, the overall running time 
to find the critical events in $[a,b]$,
ignoring the time
for computing $I_a$ and $I_b$, is $O(n^2 + K')$.

It remains to show that we can quickly find 
$I_a$ and $I_b$. We describe the algorithm for
$I_a$. 
First, compute the arrangement $\mathcal{A}_a$ of circles
with radius $a$ around the vertices of $Q$. This takes
$O(n^2)$ time~\cite{ChazelleLe86}.
To find the intersection order,
traverse in $\mathcal{A}_a$ the zone
of the line $\ell$ spanned by $e$. The time for the
traversal is bounded by the complexity of the zone.
Since the circles pairwise intersect at most twice
and $\ell$ intersects each circle only twice,
the complexity of the zone is
$O(n 2^{\alpha(n)})$~\cite[Theorem 5.11]{SharirAg95}.
Summing over all edges $e$, this adds a total of
$O(n^2 2^{\alpha(n)})$ to the running time.
To find $I_b$, we proceed similarly with $\mathcal{A}_b$.
Thus the overall time is
$O(T(n) \log(n) + n^2 2^{\alpha(n)} + K')$.
The event $K' > 8 n \ln n$ has probability less than
$1/n^4$, and we always have $K' = O(n^3)$. Thus,
this case adds $o(1)$ to the expected running time.
Given $K' \leq 8 n \ln n$, the running time
is $O(n \log n)$.
Lemma~\ref{lem:opt} follows.
Theorem~\ref{thm:opt} now results from
Lemma~\ref{lem:opt}, Theorem~\ref{thm:decpointer}, and
Theorem~\ref{thm:decram}.

\begin{theorem}\label{thm:opt}
The \frechet distance of two polygonal 
curves with $n$ edges each can be 
computed by a randomized algorithm in
time $O(n^2 \sqrt{\log n}(\log\log n)^{3/2})$ 
on a pointer machine and in time 
$O(n^2(\log\log n)^2)$ on a word RAM.
\end{theorem}

\section{Discrete \frechet Distance on the Pointer Machine}
\label{sec:discrete}

As mentioned in the introduction, 
Agarwal~\etal~\cite{AgarwalBAKaSh14}
give a subquadratic algorithm for finding
the \emph{discrete} \frechet distance
between two point sequences, 
using the word RAM. 
In this section, we explain
how their algorithm for the decision version
of the problem can be adapted to the pointer
machine. This shows that, at least for
the decision version, the speed-up does not
come from bit-manipulation tricks but from
a deeper understanding of the underlying 
geometric structure.
Our presentation is slightly different from
Agarwal~\etal~\cite{AgarwalBAKaSh14}, in order
to allow for a clearer comparison with our
continuous algorithm. 

We recall the problem definition:
we are given
two sequences $P = \langle p_1, p_2, 
\dots, p_n\rangle$ and $Q = \langle q_1, q_2, 
\dots, q_n\rangle$ of $n$ points in
the plane. For $\delta > 0$,
we define a directed graph $G_\delta$ with 
vertex set $P \times Q$. In $G_\delta$,
there is an edge between two vertices 
$(p_i, q_j)$, $(p_i, q_{j+1})$ if and only if
both $d(p_i, q_j) \leq \delta$ and 
$d(p_i, q_{j+1}) \leq \delta$. The 
condition is similar for an edge between 
vertices $(p_i, q_j)$ and $(p_{i+1}, q_j)$, 
and vertices $(p_i, q_j)$ and $(p_{i+1}, q_{j+1})$.  
There are no further edges in $G_\delta$. 
The \emph{discrete \frechet distance} 
between $P$ and $Q$ is the smallest 
$\delta$ for which $G_\delta$  has a 
path from $(p_1, q_1)$ to $(p_n, q_n)$.
In the decision version of the problem,
we are given $\delta > 0$, and we need to
decide whether there is a path from 
$(p_1, q_1)$ to $(p_n, q_n)$ in $G_\delta$.

\begin{figure}
\centering
\includegraphics{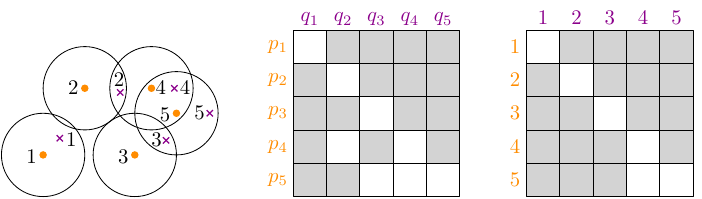}
\caption{The discrete \frechet distance:
(left) two point sequences $P$ (disks) and $Q$
(crosses) with $5$ points each;
(middle) the associated free-space matrix $F$ ($\text{white }= 1,
\text{gray } = 0$);
(right) the resulting reachability matrix
$M$. Since $M_{55} = 1$, the discrete
\frechet distance is
at most $1$.}
\label{fig:discfrechet}
\end{figure}

We now describe a subquadratic pointer machine
algorithm for the decision version. Thus,
let point sequences $P$, $Q$ be given, and
suppose without loss of generality that $\delta = 1$.
The discrete analogue of
the free-space diagram is an 
$n \times n$ Boolean matrix $F$
where $F_{ij} = 1$, if $d(p_i, q_j) \leq 1$,
and $F_{ij} = 0$, otherwise, for $i, j = 1, \dots, n$.
We call $F$ the \emph{free-space matrix}.
Similarly, the discrete analogue of the
reachable region is an $n \times n$
Boolean matrix $M$ that is defined recursively 
as follows: $M_{11} = F_{11}$, and
for $i,j = 1, \dots, n$, $(i,j) \neq 1$,
we have $M_{ij} = 1$ if and only if
$F_{ij} = 1$ and at least one of
$M_{i-1,j}$, $M_{i,j-1}$ or $M_{i-1, j-1}$
equals $1$ (we set $M_{i,0} = M_{0, j} = 0$, for
$i, j = 1, \dots n$).
Then the discrete \frechet distance between
$P$ and $Q$ is at most $1$ if and only
if  $M_{nn} = 1$. We call $M$ the 
\emph{reachability matrix};
see Figure~\ref{fig:discfrechet}.

Adapting the method of Agarwal~\etal~\cite{AgarwalBAKaSh14},
we show how to use preprocessing and table lookup in
order to decide whether $M_{nn} = 1$ in $o(n^2)$ steps on
a pointer machine.
Let $\tau = \lambda \log n$, for a suitable constant $\lambda > 0$.
We subdivide the rows of $M$ into $k = O(n/\tau)$ \emph{strips},
each consisting of $\tau$ consecutive rows: the first strip
$L_1$ consists of rows $1, \dots, \tau$, the second strip $L_2$ 
consists of rows $\tau + 1, \dots, 2\tau$, and so on.
Each strip $L_i$, $i = 1,\dots, k$,  corresponds 
to a contiguous subsequence $P_i$
of $\tau$ points on $P$. Let $\mathcal{A}_i$ be the arrangement of 
disks obtained by drawing a unit disk around each vertex in $P_i$.
The arrangement $\mathcal{A}_i$ has $O(\tau^2)$ faces.

Next, let $\rho = \lambda \log n/\log\log n$, with
$\lambda > 0$ as above. We subdivide
each strip $L_i$, $i = 1, \dots, k$ into $l = O(n/\rho)$
\emph{elementary boxes},
each consisting of $\rho$ consecutive columns in $L_i$.
We label the elementary boxes as $B_{ij}$, for $i = 1, \dots, k$
and $j = 1, \dots, l$. As above, an elementary box $B_{ij}$ has
corresponding contiguous subsequences $P_i$ of $\tau$
vertices on $P$ and $Q_j$ of $\rho$ vertices on $Q$.
Now, the \emph{incomplete signature} of an elementary box $B_{ij}$
consists of (i) the index $i$ of the strip that contains it;
and (ii) the sequence $f_1, f_2, \dots, f_\rho$ of
faces in the disk arrangement $\mathcal{A}_i$ that contain the
$\rho$ vertices of $Q_j$, in that order.
The \emph{full signature} of an elementary box $B_{ij}$
consists of its incomplete signature plus a sequence of
$\rho + \tau$ bits, that represent the entries in
the reach matrix $M$ directly above and to the left of $B_{ij}$.
We call these bits the \emph{reach bits}. As in the continuous
case, the information in the full signature suffices 
to determine how the reachability information propagates through
the elementary box; see Figure~\ref{fig:discstrips}.

\begin{figure}
\centering
\includegraphics[scale=0.9]{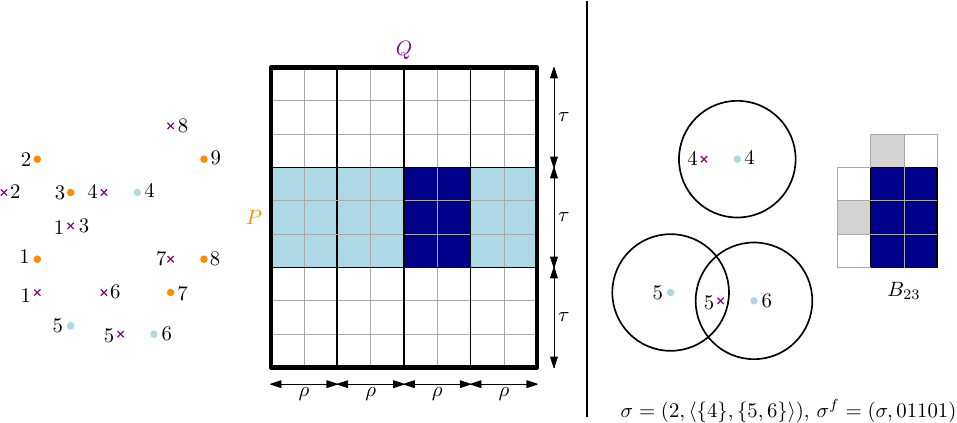}
\caption{(left) We subdivide the reachability matrix $M$
into strips of $\tau$ rows. Each strip is subdivided
into elementary boxes of $\rho$ columns. Each elementary
box corresponds to a subsequence of length $\tau$ on
$P$ and a subsequence of length $\rho$ on $Q$. (right)
The incomplete signature $\sigma$ of an elementary box 
$B_{ij}$ consists of
the index of the containing strip and the face sequence
for $Q_j$ in $\mathcal{A}_i$. The full signature $\sigma^f$
additionally contains $\rho + \tau$ reach bits that
represent the bits in $M$ directly above and to the left
of $B_{ij}$.
}
\label{fig:discstrips}
\end{figure}

The preprocessing phase proceeds as follows:
first, we enumerate all possible incomplete signatures.
For this, we need to compute all strips $L_i$ 
and the corresponding disk arrangements
$\mathcal{A}_i$, for $i = 1, \dots, \tau$. Furthermore, 
we also compute a suitable point location
structure for each $\mathcal{A}_i$. Since there are $O(n/\tau)$ 
strips, each of which consists of $\tau$ rows, this takes time
$O((n/\tau) \cdot \tau^2 \log \tau)$ = $O(n \tau \log \tau)
= O(n \log n \log\log n)$. For each strip $L_i$, 
since $\mathcal{A}_i$ has $O(\tau^2)$ faces, the number
of possible face sequences $f_1, \dots, f_\rho$ is
$\tau^{O(\rho)} \leq n^{1/3}$, by our choice
of $\tau$ and $\rho$ and for $\lambda$ small enough.
Thus, there are $O(n^{4/3}/\log n)$ incomplete signatures,
and they can be enumerated in the same time.
Now, for each incomplete signature 
$\sigma = (i, \langle f_1, \dots, f_\rho\rangle)$
we build a lookup-table that encodes for each
possible setting of the reach bits the resulting 
reach bits at the bottom and the right boundary of 
the elementary box. There are $2^{\rho + \tau} \leq n^{1/3}$
possible settings of the reach bits, by our choice of
$\tau$ and $\rho$ and for $\lambda$ small enough. 
We enumerate
all of them and organize them as a complete binary
tree of depth $\rho + \tau$. For each setting of the reach bits, we use the
information of the incomplete signature to determine 
the result through a straightforward dynamic
programming algorithm~\cite{EiterMannila94,AgarwalBAKaSh14,BringmannMu16}
in $O(\tau \cdot \rho) = O(\log^2 n)$ time, and we store the
result as a linked list of length $\rho + \tau - 1$ at the
leaf for the corresponding reach bits; see Figure~\ref{fig:discloopup}.
Thus, the total time for this part of the preprocessing
phase is $O(n^{5/3}\log n)$.

\begin{figure}
\centering
\includegraphics[scale=0.9]{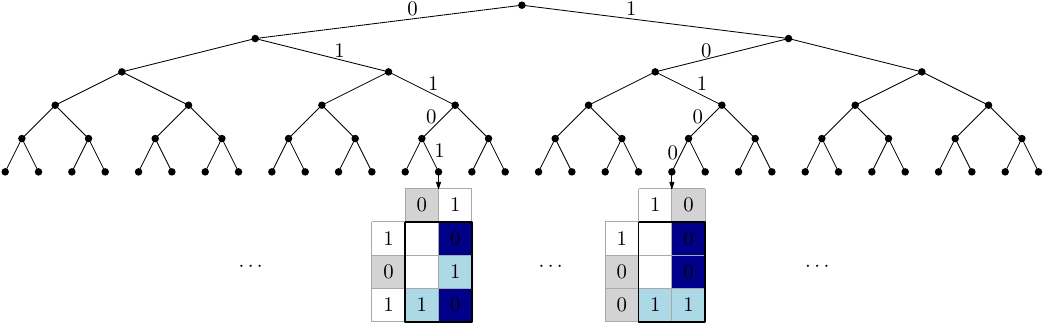}
\caption{For each incomplete signature, we create
a lookup table organized as a complete binary tree.
Each leaf corresponds to a setting of the 
reach bits for the elementary box.
In the leaves, we store a linked list of length
$\rho + \tau -1$ that represents the contents
of the reachability matrix at the bottom and
at the right of the elementary box.}
\label{fig:discloopup}
\end{figure}

Next, we determine for each elementary box $B_{ij}$
its incomplete signature. For this, we use the
point location structure for $\mathcal{A}_i$ to determine
for each vertex in $Q_j$ the face of $\mathcal{A}_i$ that
contains it. There are $O(n^2/\tau \rho)$
elementary boxes, each $Q_j$ has $\rho$
vertices, and one point location query takes
$O(\log \tau)$ time, so the total time for
this step is $O((n^2/\tau\rho)\cdot \rho \cdot \log \tau)
= O((n^2/\log n)\log\log n)$.
Using this information, we can store with each
elementary box a pointer to the lookup table
for the corresponding incomplete signature.

Finally, we can now use the lookup tables to 
propagate the reachability information through
$M$, one elementary box at a time,
as in the continuous case. The time
to process one elementary box
is $O(\rho + \tau)$, because we need to
traverse the corresponding lookup table
to find the reach bits for the adjacent
boxes. Thus, the total running time
is $O((n^2/\tau\rho) \cdot (\rho + \tau))
= O(n^2/\rho) = O((n^2/\log n) \log\log n)$.
Thus, we get the following pointer machine
version of the result by Agarwal~\etal~\cite{AgarwalBAKaSh14}

\begin{theorem}
There is an algorithm that solves
the decision version of the discrete
\frechet problem in $O((n^2/\log n) \log\log n)$
time on a pointer machine.
\end{theorem}

\paragraph{Remark}
Agarwal~\etal~\cite{AgarwalBAKaSh14} further describe
how to get a faster algorithm for the decision version 
by aggregating the elementary boxes into larger
\emph{clusters}, similar to the method 
given in Section~\ref{sec:wordram}. This improved 
algorithm finally leads to a subquadratic algorithm
for computing the discrete \frechet distance.
Unfortunately, as in Section~\ref{sec:wordram}, 
it seems that this improvement crucially
relies on constant time table lookup, so it does not
directly translate to the pointer machine.

The reader may also notice that in this section
we could choose $\tau, \rho \approx \log n$, whereas
in the previous sections we had $\tau \approx \sqrt{\log n}$.
This is due to the slightly different definition of
signature: in the discrete case, once the subsequence
$P_i$ is fixed, there are only $\tau^{O(\rho)}$ possible
ways how the subsequence $Q_j$ might interact with $P_i$.
In the continuous case, this does not seem to be so
clear, and we work with the weaker bound of $\tau^{O(\tau^2)}$
possible interactions.

\section{Decision Trees}

Our results also have implications for 
the decision-tree complexity of the 
\frechet problem. Since in that model 
we account only for comparisons between 
input elements, the preprocessing comes 
for free, and hence the size of the 
elementary boxes can be increased.
Before we consider the continuous 
\frechet problem, we first note that
a similar result can be obtained easily for the
discrete \frechet problem. 

\begin{theorem}\label{thm:actDiscrete}
The \emph{discrete} \frechet problem has
an algebraic computation tree 
of depth $\widetilde{O}(n^{4/3})$.
\end{theorem}

\begin{proof}
First, we consider the decision version:
we are given two sequences $P = p_1, \dots, p_n$
and $Q = q_1, \dots, q_n$ of $n$ points in the plane,
and we would like to decide whether the discrete
\frechet distance between $P$ and $Q$ is at most $1$.
Katz and Sharir~\cite{katz1997expander} showed that we
can compute a representation of the set of pairs
$(p_i, q_j)$ with $\|p_i - q_j\| \leq 1$ in
$\widetilde{O}(n^{4/3})$ steps. This information
suffices to complete the reachability matrix without further comparisons.
As shown by Agarwal~\etal~\cite{AgarwalBAKaSh14}, 
one can then solve the optimization problem at 
the cost of another $O(\log n)$-factor, which is 
absorbed into the $\widetilde{O}$-notation.
\end{proof}

Given our results above, we prove an analogous statement for the
continuous \frechet distance.

\begin{theorem}
There exists an algebraic decision tree 
for the \frechet problem (decision version) 
of depth $O(n^{2 - \eps})$, for a fixed constant $\eps > 0$.
\end{theorem}

\begin{proof}
We reconsider the steps of our algorithm. 
The only phases that actually involve the 
input are the second preprocessing phase and
the traversal of the elementary boxes.
The reason of our choice for $\tau$ was to 
keep the time for the first preprocessing 
phase small. This is no longer a problem.
By Lemmas~\ref{lem:incomplete} and \ref{lem:procFSD},
the remaining cost is bounded by
$O(n\tau^{5} + n^2 (\log \tau)/\tau)$. 
Choosing $\tau = n^{1/6}$,
we get a decision tree of depth
$n \cdot n^{5/6} + n^{2-1/6}\log n$. This is 
$O(n^{2-(1/6)}\log n) = O(n^{2 - \eps})$,
for any fixed $ 0 < \eps < 1/6$.
\end{proof}

\section{Weak \frechet distance}\label{sec:weak}

The \emph{weak} \frechet distance is a variant
of the \frechet distance where we are allowed
to walk \emph{backwards} along the curves~\cite{AltGo95}. 
More precisely,
let $P$ and $Q$ be two polygonal curves,
each with $n$ edges, and let $\Psi'$ be the set of 
all continuous functions $\alpha\colon [0,1] \rightarrow [0,n]$ 
with $\alpha(0) = 0$ and $\alpha(1) = n$.
The \emph{weak \frechet distance}
between $P$ and $Q$ is defined as
\[
d_\text{wF}(P,Q) \eqdef
\inf_{\alpha, \beta \in \Psi'} 
\max_{x \in [0,1]} \| P(\alpha(x)) - Q(\beta(x)) \|.
\]
Compared to the regular \frechet distance, the set 
$\Psi'$ now also contains non-monotone functions.
The weak \frechet distance was also introduced 
by Alt and Godau~\cite{AltGo95}, who showed
how to compute it in $O(n^2 \log n)$ worst-case time. 
We will now use our framework to
obtain an algorithm that runs in $o(n^2)$ expected
time on a word RAM.

\paragraph{A Decision Algorithm for the Pointer Machine}
As usual, we start with the decision version: 
given two polygonal curves $P$ and $Q$, each with
$n$ edges, decide whether $d_\text{wF}(P, Q) \leq 1$.
This has an easy interpretation 
in terms of the free-space diagram. 
Define an undirected
graph $G = (V, E)$ with vertex set $V = \{0, 1, \dots, n-1\}^2$.
The vertex $(i,j) \in V$ corresponds to the
cell $C(i,j)$, and there is an edge between
two vertices $(i, j)$ and $(i', j')$ if and only
if the two cells $C(i,j)$ and $C(i', j')$ are 
neighboring (i.e., if $|i - i'| + |j - j'| = 1$) and 
the door between them is open. Then, 
$d_\text{wF}(P, Q) \leq 1$ if and only if (i) $|P(0) - Q(0)| \leq 1$;
(ii) $|P(n) - Q(n)| \leq 1$; and (iii) the vertices 
$(0,0)$ and $(n-1, n-1)$ are in
the same connected component of $G$, see Figure~\ref{fig:weak_freespace}.
\begin{figure*}
\centering
\includegraphics{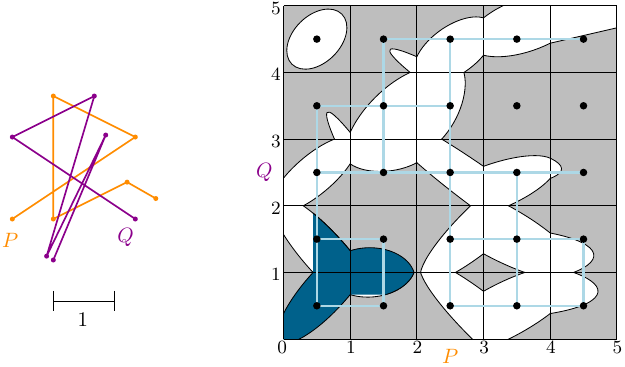}
\caption{The polygonal curves $P$ and $Q$ have
weak \frechet distance at most $1$, but \frechet distance
larger than $1$: the point $(n,n)$ is not in 
$\reach(P, Q)$, but the vertices for the
cells $C(0, 0)$ and $C(4,4)$ are in the same
connected component of $G$. The reachable region is
shown dark blue, the edges if $G$ are shown light blue.  }
\label{fig:weak_freespace}
\end{figure*}

Let $\tau, \rho \in \N$ be 
parameters, to be determined later.
We subdivide the cells into $k = O(n/\tau)$ 
\emph{vertical strips} $L_1, \dots, L_k$, 
each consisting of $\tau$ consecutive 
columns. Each strip $L_i$ corresponds to a subcurve 
$P_i$ of $P$ with $\tau$ edges. For each
such subcurve $P_i$, we define two 
arrangements $\mathcal{A}_i$ and 
$\mathcal{B}_i$. To obtain $\mathcal{A}_i$, 
we take for each edge $e$ of $P_i$ the 
``stadium'' $c_e$ of points with distance exactly $1$ 
from $e$, and we compute the resulting 
arrangement. Since two distinct curves $c_e, c_{e'}$
cross in $O(1)$ points, the complexity
of $\mathcal{A}_i$ is $O(\tau^2)$, see Figure~\ref{fig:weak_arrangement}.
The 
arrangement $\mathcal{B}_i$ is the arrangement
$\mathcal{B}$ described in the proof of
Lemma~\ref{lem:preprocess_input}, i.e.,
the arrangement of the curves dual
to the tangent lines for the unit circles 
around the vertices of $P_i$. 
\begin{figure}
\centering
\includegraphics{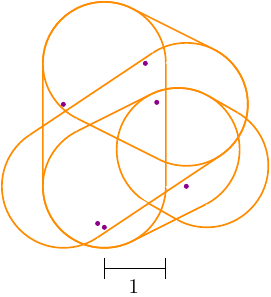}
\caption{The arrangement $\mathcal{A}_i$ for 
a subcurve $P_i$ and the some vertices of $Q$ in it.}
\label{fig:weak_arrangement}
\end{figure}

Next, we subdivide each strip into 
$\ell = O(n/\rho)$ \emph{elementary boxes},
each consisting of $\rho$ consecutive rows. 
We label the elementary boxes as
$B_{ij}$, with $1 \leq i \leq k$, 
$1 \leq j \leq \ell$.  The rows of an 
elementary box $B_{ij}$ correspond to a 
subcurve $Q_j$ of $Q$ with $\rho$ edges. 
The \emph{signature} of $B_{ij}$
consists of (i) the index $i$ of the 
corresponding strip; (ii) for each 
vertex of $Q_j$ the face of 
$\mathcal{A}_i$ that contains it; and 
(ii) for each edge $e$ of $Q_j$ the face of 
$\mathcal{B}_i$ that contains the point 
that is dual to the supporting line of $e$,
plus two indices $a,b \in \{1, \dots, \tau\}$
that indicate the first and the last
unit circle around a vertex of $P_i$ that
$e$ intersects, as we walk from one endpoint to another.

Given an elementary box $B$, the 
\emph{connection graph} $G_B$ of $B$
has $\tau\rho$ vertices, one for each
cell in $B$, and an edge between two
cells $C$, $C'$ of $B$ if and only
if $C$ and $C'$ share a (horizontal
or vertical) edge with an open door.
The \emph{connectivity list} 
of $B$ is a linked list 
with $2\tau + 2\rho - 4$ entries that 
stores for each cell $C$ on the boundary
of $B$ a pointer to a record that 
represents the connected component of 
the connection graph $G_B$ that contains $C$.

\begin{lemma}\label{lem:weak_signature}
There are $O(n\tau^{8\rho + 1})$ different 
signatures.  The connection graph 
of an elementary box $B_{ij}$ depends only on
its signature, and the connectivity list
can be computed in $O(\tau \rho)$ time on 
a pointer machine, given the signature.
\end{lemma}

\begin{proof}
First, we count the signatures. There are 
$O(n/\tau)$ strips. Once the 
strip index $i$ is fixed, a signature 
consists of $\rho+1$ faces of $\mathcal{A}_i$,
$\rho$ faces of $\mathcal{B}_i$, and $\rho$
pairs of indices $a,b \in \{1, \dots, \tau\}$.
The arrangement $\mathcal{A}_i$ has $O(\tau^2)$ 
faces, and the arrangement $\mathcal{B}_i$ has 
$O(\tau^4)$ faces, as explained in the proof 
of Lemma~\ref{lem:preprocess_input}. Finally,
there are $\tau^2$ pairs of indices. Thus,
the number of possible signatures in one strip 
is $\tau^{8\rho + 2}$.
In total, we get $O(n\tau^{8\rho+1})$
signatures. 

Next, the connection graph of an elementary box 
$B_{ij}$ is determined solely by which doors 
are open and which doors are closed. We explain
how to deduce this information from the signature. 
A horizontal edge of $B_{ij}$
corresponds to an edge $e$ of $P_i$ and a
vertex $q$ of $Q_j$. The door is open if and 
only if $q$ has distance at most $1$ from $e$.
This is determined by the face of $\mathcal{A}_i$ 
containing $q$.
Similarly, a vertical edge of $B_{ij}$ corresponds to 
a vertex $p$ of $P_i$ and an edge $e$ of $Q_j$. 
The door is open if and only if $e$ intersects
the unit circle with center $p$.
As in Lemma~\ref{lem:preprocess_input}, this can 
be inferred from the face of $\mathcal{B}_i$
that contains the point dual to the supporting line 
of $e$, together with the indices $(a,b)$ of the first
and last circle intersected by $e$.

Finally, given the signature, we can build the 
connection graph $G_{B_{ij}}$ in $O(\tau\rho)$
time, assuming that the arrangements $\mathcal{A}_i$
and $\mathcal{B}_i$ provide suitable data structures.
With $G_{B_{ij}}$ at hand, the connection list can
be found in $O(\tau\rho)$ steps, using breadth first
search.
\end{proof}

As usual, our strategy now is
to preprocess all possible signatures and to 
determine the signature of each elementary box. 
Using this information, we can then  process the
elementary boxes quickly in our main algorithm.
The next lemma describes the preprocessing steps. 

\begin{lemma}\label{lem:weak_preprocess}
We can determine for each elementary box $B_{ij}$ a 
pointer to its connectivity list in total time 
$O(n\tau^{8\rho + 2}\rho + (n^2/\tau)\log \tau)$ 
on a pointer machine.
\end{lemma}

\begin{proof}
First, we compute the arrangements $\mathcal{A}_i$ and
$\mathcal{B}_i$ for each vertical strip. By 
Lemma~\ref{lem:preprocess_input}, this takes $O(\tau^6)$
steps per strip, for a total of 
$O(\tau^6 \cdot n/\tau) = O(n\tau^5)$.
Then, we enumerate all possible signatures, and we compute
the connectivity list for each of them. By 
Lemma~\ref{lem:weak_signature}, this needs 
$O(n\tau^{8\rho + 2} \rho)$ 
time. We store the signatures and their connectivity lists

Next, we determine the signature for each elementary
box. Fix a strip $L_i$. For each vertex $q$ and each
edge $e$ of $Q$, we determine the containing faces 
of $\mathcal{A}_i$ and $\mathcal{B}_i$ and the pair
$(a,b)$ that represents the first and least 
intersection of $e$ with the unit circles around
the vertices of $P_i$. This takes $O(n\log \tau)$
steps in total, using appropriate point location
structures for the arrangements. Now, we use
the data structure from the preprocessing to 
connect each elementary box to its connectivity list.
A simple pointer-based structure supports one lookup
in $O(\rho\log \tau)$ time. Since there are
$O(n/\rho)$ elementary boxes in one strip, we get
a total running time of $O(n \log \tau)$ per strip.
Since there are $O(n/\tau)$ strips, the resulting
running time is $O((n^2/\tau)\log \tau)$.
\end{proof}

With the information from the preprocessing phase,
we can easily solve the decision problem with a
union-find data structure.

\begin{lemma}\label{lem:weak_processpointer}
Suppose that each elementary box  has a pointer
to its connectivity list. Then we can decide whether
$d_{wF}(P, Q) \leq 1$ in time 
$O\big(\frac{n^2(\tau + \rho)\alpha(n)}{\tau\rho}\big)$ 
on a pointer machine, where $\alpha(\cdot)$ denotes
the inverse Ackermann-function.
\end{lemma}

\begin{proof}
Let $S$ be the set that contains the boundary cells
of all elementary boxes. Then, 
$|S| = O\big(\frac{n^2(\tau + \rho)}{\tau\rho}\big)$.
We create a \textsc{Union}-\textsc{Find} data structure for $S$.
Then, we go through all elementary boxes $B_{ij}$, and for each
$B_{ij}$, we use the connectivity list to connect those subsets 
of boundary cells that are connected inside $B_{ij}$. This
can be done with $O(\tau + \rho)$ \textsc{Union}-operations.
Hence, the total running time for this step is 
$O\big(\frac{n^2(\tau + \rho)}{\tau\rho}\big)$. (We do not need any 
\textsc{Find}-operations yet, because we can store
the representatives of the sets with the representatives of
the connected components in the connectivity list as we walk
along the boundary of an elementary box.)

Next, we iterate over the boundary cells of all elementary
boxes. For each boundary cell $C$, we determine the neighboring
boundary cells on neighboring  elementary boxes that share
an open door with $C$. For each such cell $D$, we perform 
\textsc{Find}-operations on $C$ and $D$, followed by a 
\textsc{Union}-operation.
This takes $O\big(\frac{n^2(\tau + \rho)\alpha(n)}{\tau\rho}\big)$
time~\cite{Tarjan75}. Finally, we return \textsc{True} if and only if
(i) the pairs of start and end vertices of $P$ and $Q$ have 
distance at most $1$, and (ii) $C(0,0)$ and $C(n-1, n-1)$ are
in the same set of the resulting partition of $S$.
The running time follows.
\end{proof}

The following theorem summarizes our algorithm
for the decision problem.

\begin{theorem}\label{thm:weak_pointer}
Let $P$ and $Q$ be two polygonal curves with
$n$ edges each. We can decide whether
$d_{wF}(p, q) \leq 1$ in $O((n^2/\log n)\alpha(n)\log\log n)$ time 
on a pointer machine, where $\alpha(\cdot)$ denotes the
inverse Ackermann function.
\end{theorem}

\begin{proof}
We set $\tau = \log n$ and $\rho = \lambda \log n/\log\log n$,
for a suitable constant $\lambda > 0$. If $\lambda$ is
small enough, then $\tau^{8\rho + 2} = O(n^{4/3})$, and
the algorithm from Lemma~\ref{lem:weak_preprocess} runs in
time $O((n^2/\log n)\log\log n)$. After the preprocessing
is finished, we can use Lemma~\ref{lem:weak_processpointer}
to obtain the final result in time $O((n^2/\log n)\alpha(n)\log\log n)$.
\end{proof}

\paragraph{A Faster Decision Algorithm on the Word RAM}
Theorem~\ref{thm:weak_pointer} is 
too weak to obtain a subquadratic algorithm for the
weak \frechet distance. This requires the full power of the
word RAM.
\begin{theorem}
Let $P$ and $Q$ be two polygonal curves with
$n$ edges each. We can decide whether
$d_{wF}(p, q) \leq 1$ in $O((n^2/\log^2 n)(\log\log n)^5)$ time 
on a word RAM.
\end{theorem}

\begin{proof}
We set $\tau = \lambda \log^2 n / \log\log n$ and 
$\rho = \lambda \log n/ \log\log n$.
If $\lambda$ is small enough, then $\tau^{8\rho + 2} = O(n^{4/3})$, and
we can perform the  algorithm from Lemma~\ref{lem:weak_preprocess} in
time $O((n^2/\log^2 n)\log\log n)$.

We modify the algorithm from Lemma~\ref{lem:weak_preprocess}
slightly. Instead of a pointer-based connectivity list,
we compute  a \emph{packed connectivity list}.
It consists of $O(\log n)$ words 
of $\log n$ bits each. Each word stores $\Theta(\log n / \log \log n)$
\emph{entries} of $O(\log\log n)$ bits. An entry consists of
$O(1)$ \emph{fields}, each with $O(\log \log n)$ bits.
As in Lemma~\ref{lem:weak_preprocess}, the entries in the
packed connectivity list represent the connected components of
the connection graph $G_B$ for the
boundary cells of a given elementary box $B$.
This is done as follows: each entry of the connectivity list
corresponds to a boundary cell $C$ of $B$. In the first field,
we store a unique identifier from $\{1, \dots, 2\tau + 2\rho -4\}$
that identifies $C$. In the second field, we store the smallest 
identifier of any boundary cell of $B$ that lies in the same connected
component as $C$. The remaining fields of the entry are initialized 
to $0$.
Furthermore, we compute for each elementary box $B$ a sequence of
$O((\tau / \log n)\log\log n)$ words that indicates for each
boundary cell of $B$ whether the corresponding door to 
the neighboring elementary box is open. This information can
be obtained in the same time by adapting the algorithm from
Lemma~\ref{lem:weak_preprocess}.

Next, we group the elementary boxes into \emph{clusters}.
A cluster consists of $\log n$ vertically adjacent elementary
boxes from a single strip. The first set of 
clusters come from bottommost $\log n$ elementary boxes,
the second set of clusters from following $\log n$ elementary boxes,
etc. The boundary and
the connectivity list of a cluster are defined analogously as
for an elementary box. Below, in Lemma~\ref{lem:weak_cluster},
we show that we can compute the
(pointer-based) connectivity list of a cluster in time 
$O((\tau + \rho) (\log\log n)^4)$. 
Then, the
lemma follows:  there are $O(n^2/(\tau \rho \log n))$
clusters, so the total time to find connectivity lists
for all clusters is 
$O(n^2/(\rho \log n)(\log\log n)^4 + n^2/(\tau \log n)(\log\log n)^4) = 
O((n^2/\log^2 n) (\log\log n)^{5})$.
After that, we can solve the decision
problem in time
\[
O\left(\frac{n^2(\tau + \rho\log n)\alpha(n)}{\tau\rho\log n}\right) 
= 
O\left(\frac{n^2\alpha(n)\log\log n}{\log^2 n}\right), 
\]
as in Lemma~\ref{lem:weak_processpointer}.
\end{proof}

\begin{lemma}\label{lem:weak_cluster}
Given a cluster, we can compute 
the connectivity list for its boundary
in total time $O((\tau + \rho)(\log\log n)^4)$
on a word RAM.
\end{lemma}

\begin{proof}
Our algorithm makes extensive use of the 
word RAM capabilities. Data is processed in \emph{packed
form}. That is, we usually deal with a sequence of
packed words, each with $O(\log n/\log\log n)$
\emph{entries}. Each entry has $O(1)$ {fields}
of $O(\log\log n)$ bits. We restrict the number of
entries in each word so that the total number of
bits used is at most, say, $(1/3)\log n$. Then, we
can support any reasonable binary operation on two packed
words in $O(1)$ time after $o(n)$ preprocessing time.
Indeed, we only need to
build a lookup-table during the preprocessing phase.
First, we observe that we can \emph{sort} packed data efficiently.
\begin{claim}\label{clm:packed_sort}
Suppose we are given a sequence of $\mu$ packed words, 
representing a sequence of $O(\mu(\log n/\log\log n))$
entries. We can obtain a sequence of $\mu$
packed words that represents the same entries, 
sorted according to any given field, 
in $O(\mu \log \mu)$ time. 
\end{claim}

\begin{proof}
We adapt usual techniques for packed integer 
sorting~\cite{AlbersHa97}.
We precompute a \emph{merge}-operation that receives
two packed words, each with their entries sorted according to a
given field, and returns two packed words that represent the
sorted sequence of all entries. We also precompute an operation that
receives one packed word and returns a packed word with the
same entries, sorted according to a given field. Then,
we can perform a 
merge sort in $O(\mu \log \mu)$ time, because
in each level of the recursion, the total time for merging
is $O(\mu)$. Once we are down to a single 
word, we can sort its entries in one step. More details can
be found in the literature on packed 
sorting~\cite{AlbersHa97,BuchinMu11}.
\end{proof}

\begin{figure}
\centering
\includegraphics[scale=0.5]{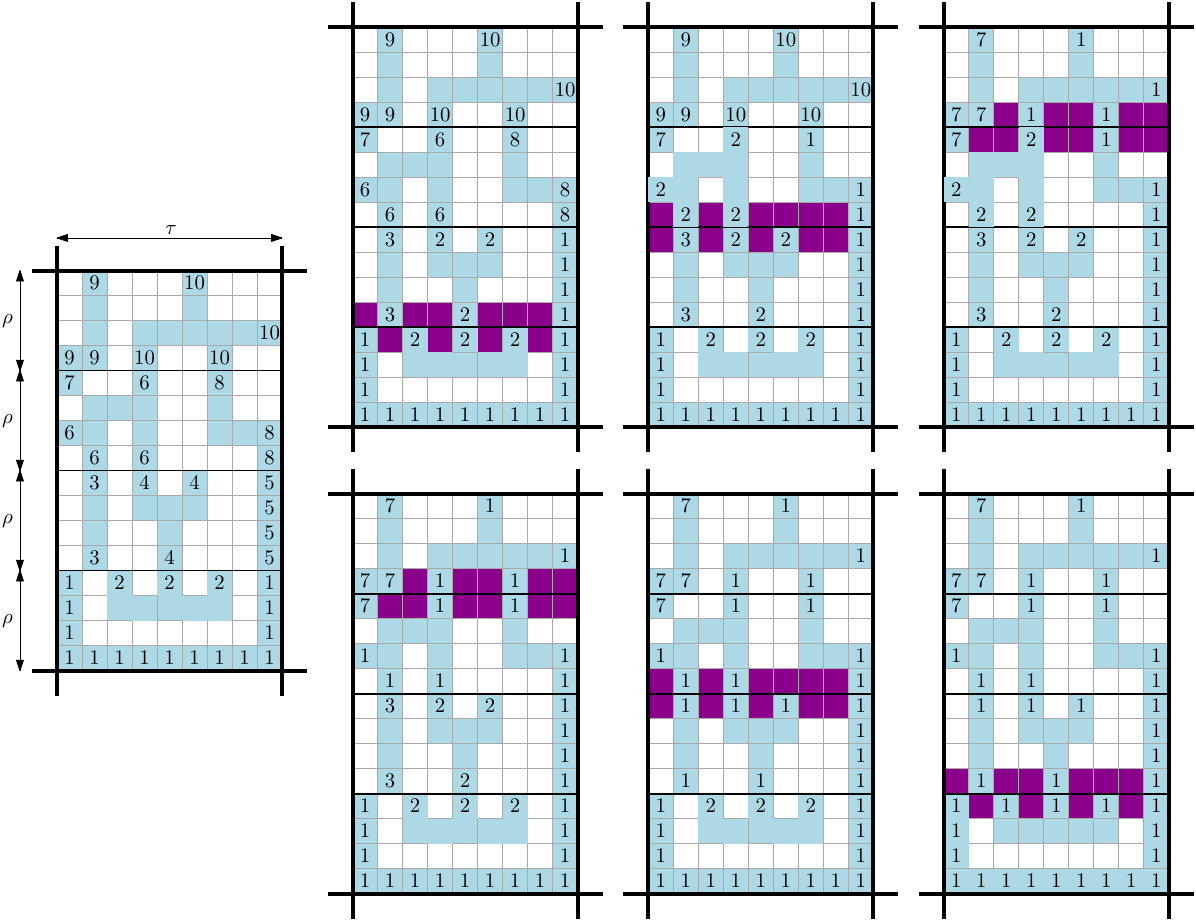}
\caption{The algorithm for computing the connected components
within a cluster. In the beginning, we know the components
inside each elementary box (left). In the first phase (top right), we 
update
the connectivity information from bottom to top, by considering
the upper and lower boundaries of the vertically adjacent elementary
boxes. In the second phase (bottom right), we propagate the connected 
components from top to bottom.}
\label{fig:weak_cluster}
\end{figure}
Now we describe the strategy of the main algorithm.
As explained above,
a cluster consists of $\log n$ vertically adjacent elementary boxes. 
We number the boxes
$B_1, \dots, B_{\log n}$, from bottom to top. For 
$i = 1, \dots, \log n$, we denote by $G_i$ the connection
graph for the boxes $B_1, \dots, B_i$. That is, $G_i$ is obtained
by taking the union $\bigcup_{j = 1}^i G_{B_j}$
of the individual connection graphs and adding 
edges for adjacent cells in neighboring boxes that share an
open door. Our algorithm proceeds
in two \emph{phases}. In the first phase, we propagate the connectivity 
information upwards. 
That is, for $i = 1, \dots, \log n$, we compute
a sequence $W_i$ of $O((\tau + \rho)\log\log n/ \log n)$ packed words. 
Each entry in $W_i$
corresponds to a cell $C$ on the boundary of $B_i$. 
The first field stores a unique identifier for the cell, and the
second entry stores an identifier of the connected component of
$C$ in $G_i$.
In the second phase, we go in the reverse direction. For 
$i = \log n, \dots, 1$,
we compute a sequence $W_i'$ of $O((\tau + \rho)\log\log n/\log n)$ 
packed words that store
for each cell $C$ on the boundary of $B_i$ an identifier of the
connected component of $C$ in $G_{\log n}$. Once this information is 
available, the (pointer-based) connectivity list for the cluster 
boundary can be extracted in $O(\tau + \rho)$
time, using appropriate precomputed operations on packed words,
see Figure~\ref{fig:weak_cluster} .

We begin with the first phase. The \emph{upper boundary}
of an elementary box are the $\tau$ cells in the topmost
row of the box, the \emph{lower boundary} are the $\tau$
cells in the bottommost row. From the preprocessing phase,
we have a pointer to the packed connectivity lists for all 
elementary boxes $B_1, \dots, B_{\log n}$.
We make local copies of these packed connectivity lists, and we 
modify them so that the identifiers of all cells and connected
components are unique. For example, we can
increase all identifiers in the connectivity list for $B_i$
by $i \log^3 n$. Using lookup-tables for these word-operations, 
this step can be carried out in $O((\tau + \rho)\log\log n)$ time for
the whole cluster.

The sequence $W_1$ for $B_1$ is exactly the packed connectivity list
of $G_1$. Now, suppose that we have computed for $B_i$ a sequence
$W_i$ of packed words that represent the connected component in $G_i$ 
for each boundary cell of $B_i$. We need to compute a 
similar sequence $W_{i+1}$ for $B_{i+1}$. 
For this, we need to determine how the doors
between $B_i$ and $B_{i+1}$ affect the connected components
in $G_{B_{i+1}}$.

Let $H = (V_H, E_H)$ be the graph that has one vertex for each 
connected component of $G_i$ that contains a cell at the upper 
boundary of $B_i$ and one vertex for each connected component of 
$G_{B_{i+1}}$ that contains a cell at the lower boundary of $B_{i+1}$. 
The vertices of $H$ are labeled by the identifiers of their corresponding
components.  Suppose that $v \in V_H$ represents a component $D_v$ in
$G_{i}$ and $w \in V_H$ represents a component $D_w$ in $G_{B_{i+1}}$. 
There is an edge between $v$ and $w$ in $E_H$ if and only if
$D_v$ contains a cell on the upper boundary of $B_{i}$ and
$D_w$ contains a cell on the lower boundary of $B_{i+1}$ such
that these two cells share an open door.
The graph $H$ has $O(\tau)$ vertices and $O(\tau)$ edges.
Our goal is to obtain the connected components of $H$. From this,
we will then derive the next sequence $W_{i+1}$.
To obtain the components of $H$, we adapt the parallel connectivity 
algorithm of Hirschberg, Chandra, and 
Sarwate~\cite{HirschbergChSa79}.

\begin{claim}\label{clm:packed_components}
We can determine the connected components of $H$ in time
$O((\tau/\log n) \log\log^4)$ time.
\end{claim}

\begin{proof}
First, we obtain a list of $O((\tau/\log n)\/\log\log n)$ packed words
that represent the vertices of $H$, as follows: we 
extract from $W_i$ and from the connectivity list of $B_{i+1}$
the entries for the upper boundary of $B_i$ and for the lower boundary
of $B_{i+1}$. If all these entries are stored together in the
respective packed lists, this takes $O((\tau/\log n)\log\log n)$ time to
copy the relevant data words. Then, we sort these packed sequences
according to the component identifier, in time 
$O(\tau/\log n)\log\log^2 n)$ (Claim~\ref{clm:packed_sort}).
Finally, we go over the sorted packed sequence to extract
a sorted lists of the distinct component identifiers.
This takes $O((\tau/\log n)\log\log n)$ time, using an appropriate
operation on packed words.

The edges of $E_H$ can also be represented by a sequence of 
$O((\tau/\log n)\log\log n)$ words. As explained above, we have
from the preprocessing phase for each elementary box a sequence
of $O((\tau/\log n)\log\log n)$ words
whose entries indicate the open doors along its upper boundary.
From this, we can obtain in $O((\tau/\log n)\log\log n)$ time a
sequence of $O((\tau/\log n))\log\log n)$ words whose entries represent 
the edges of $H$, where each entry stores the component identifiers
of the two endpoints of the edge.
We are now ready to implement the algorithm of Hirschberg, Chandra,
and Sarwate. The main steps of the algorithm are as follows,
see Figure~\ref{fig:hcs_connect}.
\begin{figure}
\centering
\includegraphics{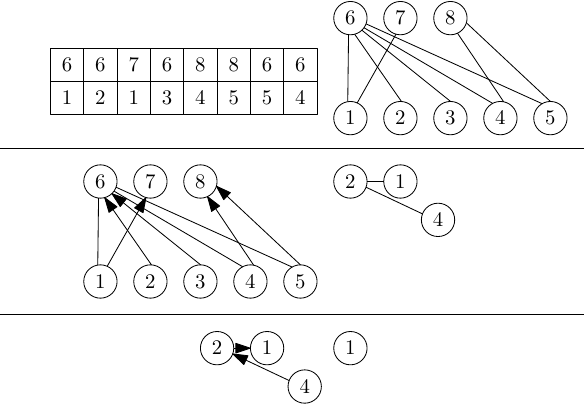}
\caption{The algorithm of Hirschberg, Chandra, and Sarwate.
In each step, we either connect all nodes to their largest
neighbor or to their smallest neighbor.
Then, we contract the resulting rooted trees and label the
contracted nodes with the index of the smallest node.}
\label{fig:hcs_connect}
\end{figure}

\noindent
\textbf{Step 1}: Find for each vertex of $H$ the neighbor with
the smallest and with the largest identifier.
This can be done in $O((\tau/\log n)\log\log^2 n)$ time
by sorting the edge lists twice, once in lexicographic order and
once in reverse lexicographic order of the identifiers of the endpoints.
From these sorted lists, we can extract the desired information in the
claimed time, using appropriate word operations.

\noindent
\textbf{Step 2}: Let $V_H'$ be the vertices
of $H$ with at least one neighbor, and let $V_H'' \subseteq V_H'$ be the
vertices $v \in V_H'$ having
a neighbor with a smaller identifier than
the identifier of $v$. If $|V_H''| \geq |V_H'|/2$, we set the
\emph{successor} of each $v \in V_H''$ to the neighbor of $v$
with the smallest identifier. Otherwise, at least
half of the nodes in $V_H'$ have a neighbor with a larger identifier.
In this case, we let $V_H'''$ be the set of these nodes, and 
we set the
\emph{successor} of each $v \in V_H'''$ to the neighbor of $v$
with the largest identifier. The successor relation defines a directed
forest $F$ on $V_H$ such that at least half of the vertices
in $V_H'$ are not a root in $F$. Given the information available
from Step 1 and appropriate word operations, this step can 
carried out in $O((\tau/\log n)\log\log n)$ time.

\noindent
\textbf{Step 3}: Use \emph{pointer jumping} to determine for each
vertex $v \in V_H$ the identifier of the root of the tree in  $F$ 
that contains $v$. For this, we set the successor of each $v \in V_H$
that does not yet have a successor to $v$ itself. Then, for 
$\log |V_H| = O(\log\log n)$ rounds, we set simultaneously for 
each $v \in V_H$ the new successor of $v$ to the old successor of the 
old successor of $v$ (pointer jumping). 
Each step at least halves the distance of $v$ to its root in $F$, so
it takes $O(\log |V_H|)$ rounds until each vertex in $V_H$
has found the root of its tree in $F$. Each round can be implemented
in $O((\tau/\log n)\log\log^2 n)$ time by sorting the vertices according
to their successors.
Thus, this step takes $O((\tau/\log n)\log\log^3 n)$ time in total.

\noindent
\textbf{Step 4}: Contract each tree of $F$ into a single vertex
whose identifier is the smallest identifier in the tree..
Maintain for the original vertices of $H$ a list that gives the
identifier of the contracted node that represents it.
Again, this step can be carried out in $O(\tau/\log n)\log\log^2 n)$ time
using sorting.

After Steps~1--4, the number of non-singleton components
in $H$ has at least halved. Thus,
by repeating the steps $O(\log |V_H|) = O(\log\log n)$ times,
we can identify the connected components of $H$.
The total time of the algorithm is 
$O((\tau/\log n)\log\log^4 n)$, as claimed.
\end{proof}

Given the connected components of $H$, we can find the
desired sequence $W_{i+1}$ for the boundary $B_{i+1}$.
Indeed, the procedure from Claim~\ref{clm:packed_components}
outputs a sequence of $O((\tau/\log n)\log\log n)$ packed words
that gives for each vertex in $H$ an identifier of the component
in $H$ that contains it. We can use this list as a lookup table
to update the identifiers of the components in the connectivity list
of $B_{i+1}$. This takes $O((\tau+\rho)/\log n)\log\log^2 n)$ time,
using sorting.

In summary, since we consider $\log n$ elementary boxes,
the total time for the first phase if $O((\tau + \rho)\log\log^4 n)$.
The second phase is much easier. For $i = \log_n, \dots, 2$, we
propagate the connectivity information from $B_{i+1}$ to $B_i$.
For this, we need
to update the identifiers of the connected components for the cells on the
upper boundary of $B_i$ using the identifiers of the connected components
on the lower boundary of $B_{i+1}$, and then adjust the connectivity
list $B_{i+1}$ with these new indices. Again, 
this takes $O((\tau+\rho)/\log n)\log\log^2 n)$ time, using sorting.
\end{proof}

\paragraph{Computing the Weak \frechet Distance}
To actually compute the weak \frechet distance, we use
a simplified version of the procedure from Section~\ref{sec:compute}.
In particular, for the weak \frechet distance, there are
only critical values of the type 
vertex-vertex and vertex-edge, i.e., there are only 
$O(n^2)$ critical values. However, we aim
for a subquadratic running time, so we need to perform the 
sampling procedure in a slightly different way.

\begin{theorem}
Suppose we can answer the decision problem for the 
weak \frechet distance in time $T(n)$, for input curves
$P$ and $Q$ with $n$ edges each. Then, we can compute the
weak \frechet distance of $P$ and $Q$ in expected time
$O(n^{3/2}\log^c n + T(n)\log n)$,
for some fixed constant $c > 0$.
\end{theorem}
\begin{proof}
First, we sample a set $S$ of $K = 6n^{1/2}$ critical values uniformly
at random. Then, we find $a, b$ such that the weak \frechet distance
lies between $a$ and $b$ and such that the interval $[a, b]$ contains 
no other element from $S$. This takes $O(K + T(n) \log n)$ time, using 
median finding. 

Similarly to Har-Peled and 
Raichel~\cite[Lemma~6.2]{HarPeledRa14},
we see that the probability, that the interval 
$[a,b]$ has more than $2\gamma n^{3/2} \ln n$ critical values
is at most $1/n^\gamma$. 
Indeed, there are at most $n^2$ vertex-vertex and at most
$2n^2$ vertex-edge events. Thus, the total number of critical
values is at most $L \leq 3n^2$. Let 
$U^+$ be the next $\gamma n^{3/2}\ln n$ larger critical values
after $d_\text{wF}(P, Q)$.
Then, the probability that $S$ contains no value from $U^+$ is
\[
\frac{\binom{L - |U^+|}{K}}{\binom{L}{K}}
=
\prod_{i = 0}^{K-1} \left(1 - \frac{|U^+|}{L - i}\right)
\leq 
\left(1 - \frac{|U^+|}{L}\right)^K
\leq \exp\left(- \frac{|U^+| K}{L}\right)
\leq \exp(-2\gamma\ln n) \leq \frac{1}{2n^\gamma}.
\]
Analogously, the probability that $S$ contains none of the
next $\gamma n^{3/2}\ln n$ smaller critical values is also
at most $1/2n^\gamma$, so the claim follows.
For $\gamma > 0$ large enough, the contribution of
this event to the expected running time is negligible.

Next, we find all critical values in the
interval $[a, b]$. For this, we must determine all vertex-vertex and
vertex-edge pairs with distance in $[a,b]$. 
We report for every vertex $p$ of $P$ or $Q$ the set of
vertices of the other curve that lie in the annulus with radii $a$ 
and $b$ around $p$. Furthermore, we report for every edge
$e$ of $P$ or $Q$ the set of vertices of the other curve that lie in 
the stadium with radii $a$ and $b$.
This can be done efficiently with a range-searching 
structure for semi-algebraic sets. 
Agarwal, Matou\v{s}ek and Sharir~\cite{AgarwalMaSh13} 
show that we can preprocess
the vertices of $P$ and the vertices of $Q$ into a data structure
that can answer our desired range reporting queries in time
$O(n^{1/2}\log^c n + k)$, where $k$ is the output size and $c > 0$ is
some fixed constant. The expected preprocessing time is 
$O(n^{1+\eps})$, where $\eps > 0$ can be made arbitrarily small.
We perform $O(n)$ queries, and the total expected output
size of our queries is $O(n^{3/2}\ln n)$, so it takes
expected time $O(n^{3/2}\log^{c} n)$ to find the critical values in
$[a,b]$. Finally, we perform a binary search on theses critical
values to compute the weak \frechet distance. This takes $O(T(n)\log n)$
time. 
\end{proof}

The following theorem summarizes our results on the weak
\frechet distance.
\begin{theorem}
The weak \frechet distance of two polygonal 
curves, each with $n$ edges, can be 
computed by a randomized algorithm in
time $O(n^2 \alpha(n)\log\log n)$
on a pointer machine and in time 
$O((n^2/\log n)(\log\log n)^5)$ on a word RAM.
\end{theorem}

\section{Conclusion}

We have broken the long-standing quadratic 
upper bound for the decision version of the 
\frechet problem.  Moreover, we have shown 
that this problem has an algebraic decision 
tree of depth $O(n^{2-\eps})$, for some 
$\eps > 0$ and where $n$ is the number of 
vertices of the polygonal curves. We have shown 
how our faster algorithm for the decision version 
can be used for a faster algorithm to compute 
the \frechet distance. If we allow constant-time 
table-lookup, we obtain a running time in close 
reach of $O(n^2)$.  
This leaves us with intriguing open research 
questions.  Can we devise a quadratic or even
a slightly subquadratic algorithm for the 
optimization version? Can we devise such an
algorithm on the word RAM, that is, with constant-time 
table-lookup? What can be said about approximation algorithms?

\section*{Acknowledgments}
We would like to thank Natan Rubin for 
pointing out to us that \frechet-related 
papers require a witty title involving a dog, 
Bettina Speckmann for inspiring discussions 
during the early stages of this research, 
G\"unter Rote for providing us with his Ipelet for
the free-space diagram, and Otfried Cheong for Ipe.
We would also like to thank anonymous referees for 
suggesting simpler proofs of Lemma~\ref{lem:preprocess_input} 
and Lemma~\ref{lem:opt}, for
pointing out~\cite{katz1997expander} for 
Theorem~\ref{thm:actDiscrete}, and for 
suggesting a section on the weak \frechet distance. 

\newcommand{\SortNoop}[1]{}

\appendix

\section{Computational Models}
\label{app:compmodels}

\paragraph{Real RAM}
The standard machine model in computational geometry
is the \emph{real RAM}~\cite{PreparataSh85}. Here, data is represented
as a (countably) infinite sequence of storage cells. These cells
can be of two different types: they can store real numbers
or integers. The model supports standard operations on
these numbers in constant time, including addition,
multiplication, and elementary functions
like square-root, sine or cosine. Furthermore,
the integers can be used as indices to memory locations.
Integers can be converted to real numbers in constant time,
but we need to be careful about the reverse direction.
The \emph{floor} function can be used to truncate a
real number to an integer, but if we were allowed to
use it arbitrarily, the real RAM could solve PSPACE-complete
problems in polynomial time~\cite{Schoenhage79}.
Therefore, we usually have only a restricted floor function
at our disposal.

\paragraph{Word RAM}
The \emph{word RAM} is essentially a real RAM without
support for real numbers. However, on a real RAM,
the integers are usually treated as atomic, whereas
the word RAM allows for powerful bit-manipulation 
tricks~\cite{FredmanWi93}.
More precisely, the word RAM represents the
data as a sequence of
$w$-bit words, where $w = \Omega(\log n)$.
Data can be accessed arbitrarily, and standard
operations, such as Boolean operations
(\texttt{and}, \texttt{xor}, \texttt{shl}, $\ldots$), addition, or
multiplication take constant time. There are many variants of
the word RAM, depending on precisely which instructions are
supported in constant time~\cite{FredmanWi93,BuchinMu11}. 
The general consensus seems
to be that any function in $\text{AC}^0$
is acceptable.\footnote{$\text{AC}^0$ is the
class of all functions $f\colon \{0,1\}^* \rightarrow \{0,1\}^*$ that
can be computed by a family of circuits $(C_n)_{n \in \N}$ with the
following properties~\cite{AroraBa09}: (i) each $C_n$ has $n$ inputs; 
(ii) there exist constants
$a,b$, such that $C_n$ has at most $an^b$ gates, for $n\in \N$;
(iii) there is a constant $d$ such that for all $n$ the length 
of the longest
path from an input to an output in $C_n$ is at most $d$ (i.e., the
circuit family has bounded depth); (iv) each gate
has an arbitrary number of incoming edges (i.e., the fan-in is unbounded).}
However, it is always preferable to rely on a set of operations
as small, and as non-exotic, as possible.
Note that multiplication is not in $\text{AC}^0$~\cite{FurstSaSi84},
but nevertheless is often
included in the word RAM instruction set~\cite{FredmanWi93}.

\paragraph{Pointer machine}
The  \emph{pointer machine} model disallows the use of constant
time table lookup, and is therefore
a restriction of the (real) RAM model~\cite{Schoenhage80,Tarjan75}.
The data structure is modeled as a directed
graph $G$ with bounded out-degree. Each node in
$G$ represents a \emph{record}, with a bounded
number of pointers to other records and a bounded number
of (real or integer) data items.
The algorithm can access data only by following pointers
from the inputs  (and a bounded number of global entry
records); random access is not possible. The data can be manipulated
through the usual real RAM operations,
but without support for the floor function, for
reasons mentioned above.

\paragraph{Algebraic computation tree} \emph{Algebraic computation trees}
(ACTs)~\cite{AroraBa09} are the computational geometry analogue of binary decision trees,
and like these they are mainly used for proving lower bounds. Let
$x_1, \ldots, x_n \in \R$ be the inputs. An ACT is a binary tree
with two kinds of nodes: \emph{computation} and
\emph{branch} nodes. A computation node $v$ has one child and
is labeled with an
expression of the type $y_v = y_u \oplus y_w$, where
$\oplus \in \{+,-,*,/,\sqrt{\cdot}\}$ is a operator and
$y_u, y_w$ is either an input variable $x_1, \ldots, x_n$ or
corresponds to a computation node that is an ancestor of $v$.
A branch node has degree 2 and is labeled by $y_u = 0$ or
$y_u > 0$, where again $y_u$  is either an input or a variable
for an ancestor. A family of algebraic computation trees
$(T_n)_{n \in \N}$
solves a computational problem (like Delaunay triangulation or
convex hull computation), if for each $n \in \N$, the tree $T_n$
accepts inputs of size $n$, and if for any such input $x_1, \ldots, x_n$
the corresponding path in  $T_n$ (where the children of the branch
nodes are determined according the conditions they represent)
constitutes a computation that represents the answer in the
variables $y_v$ encountered during the path.

Algebraic \emph{decision} trees
are defined as follows:
we allow only branch nodes. Each branch node is
labeled with a predicate of the form
$p(x_1, \dots, x_n) = 0$ or $p(x_1, \dots, x_n) > 0$.
The leaves are labeled \emph{yes} or \emph{no}.
Fix some $r \in \{1, \dots, n\}$.
If $p$ is restricted to be
of the form $p(x_1, \dots, x_n) = \sum_{i = 1}^n a_i x_i -b$,
with at most $r$ coefficients $a_i \neq 0$,
we call the decision tree \emph{$r$-linear}.
Erickson~\cite{Erickson99} showed that
any $3$-linear  decision tree for  3SUM
has depth $\Omega(n^2)$. However, 
Gr\o{}nlund and Pettie showed that
there is a $4$-linear decision tree
of depth $O(n^{3/2} \sqrt{\log n})$
for the problem.
In geometric problems, linear predicates
are often much too restrictive.
For example, there is
no $r$-linear decision tree for the \frechet
problem, no matter the choice of $r$: with $r$-linear
decision trees, we cannot even decide whether two
given points $p$ and $q$ have Euclidean distance
at most $1$.


\begin{thebibliography}{10}

\bibitem{AgarwalBAKaSh14}
P.~K. Agarwal, R.~Ben~Avraham, H.~Kaplan, and M.~Sharir.
\newblock Computing the discrete {F}r{\'{e}}chet distance in subquadratic time.
\newblock {\em SIAM J. Comput.}, 43(2):429--449, 2014.

\bibitem{AgarwalHMS05}
P.~K. Agarwal, S.~Har-Peled, N.~H. Mustafa, and Y.~Wang.
\newblock Near-linear time approximation algorithms for curve simplification.
\newblock {\em Algorithmica}, 42(3--4):203--219, 2005.

\bibitem{AgarwalMaSh13}
P.~K. Agarwal, J.~Matou\v{s}ek, and M.~Sharir.
\newblock On range searching with semialgebraic sets. {II}.
\newblock {\em SIAM J. Comput.}, 42(6):2039--2062, 2013.

\bibitem{AilonCh05}
N.~Ailon and B.~Chazelle.
\newblock Lower bounds for linear degeneracy testing.
\newblock {\em J. ACM}, 52(2):157--171, 2005.

\bibitem{AlbersHa97}
S.~Albers and T.~Hagerup.
\newblock Improved parallel integer sorting without concurrent writing.
\newblock {\em Inform. and Comput.}, 136(1):25--51, 1997.

\bibitem{Alt09}
H.~Alt.
\newblock The computational geometry of comparing shapes.
\newblock In {\em Efficient Algorithms}, volume 5760 of {\em Lecture Notes in
  Computer Science}, pages 235--248. Springer-Verlag, 2009.

\bibitem{AltBuchin10}
H.~Alt and M.~Buchin.
\newblock Can we compute the similarity between surfaces?
\newblock {\em Discrete Comput. Geom.}, 43(1):78--99, 2010.

\bibitem{AltGo95}
H.~Alt and M.~Godau.
\newblock Computing the {F}r{\'e}chet distance between two polygonal curves.
\newblock {\em Internat. J. Comput. Geom. Appl.}, 5(1--2):78--99, 1995.

\bibitem{AltKW03}
H.~Alt, C.~Knauer, and C.~Wenk.
\newblock Comparison of distance measures for planar curves.
\newblock {\em Algorithmica}, 38(1):45--58, 2003.

\bibitem{AronovHKWW06}
B.~Aronov, S.~Har-Peled, C.~Knauer, Y.~Wang, and C.~Wenk.
\newblock {Fr{\'e}chet Distances for Curves, Revisited}.
\newblock In {\em Proc. 14th Annu. European Sympos. Algorithms (ESA)}, volume
  4168 of {\em Lecture Notes in Computer Science}, pages 52--63, 2006.

\bibitem{AroraBa09}
S.~Arora and B.~Barak.
\newblock {\em Computational complexity. A modern approach}.
\newblock Cambridge University Press, Cambridge, 2009.

\bibitem{BaranDP08}
I.~Baran, E.~D. Demaine, and M.~P\v{a}tra\c{s}cu.
\newblock Subquadratic algorithms for {3SUM}.
\newblock {\em Algorithmica}, 50(4):584--596, April 2008.

\bibitem{BellmanKalaba59}
R.~Bellman and R.~Kalaba.
\newblock On adaptive control processes.
\newblock {\em IRE Transactions on Automatic Control}, 4(2):1--9, 1959.

\bibitem{BenAvFiKaKaSh15}
R.~{Ben Avraham}, O.~Filtser, H.~Kaplan, M.~J. Katz, and M.~Sharir.
\newblock The discrete and semicontinuous {F}r{\'{e}}chet distance with
  shortcuts via approximate distance counting and selection.
\newblock {\em ACM Transactions on Algorithms}, 11(4):29:1--29:29, 2015.

\bibitem{BergCG13}
M.~{\SortNoop{Berg}}de~Berg, A.~F. {Cook IV}, and J.~Gudmundsson.
\newblock Fast {F}r{\'e}chet queries.
\newblock {\em Comput. Geom. Theory Appl.}, 46(6):747--755, 2013.

\bibitem{BrakatsoulasPSW05}
S.~Brakatsoulas, D.~Pfoser, R.~Salas, and C.~Wenk.
\newblock {On map-matching vehicle tracking data}.
\newblock In {\em Proc. 31st Int. Conf. on Very Large Data Bases}, pages
  853--864, 2005.

\bibitem{BremnerChDeErHuIaLaTa12}
D.~Bremner, T.~M. Chan, E.~D. Demaine, J.~Erickson, F.~Hurtado, J.~Iacono,
  S.~Langerman, M.~P\v{a}tra\c{s}cu, and P.~Taslakian.
\newblock Necklaces, convolutions, and {$X+Y$}.
\newblock {\em Algorithmica}, pages 1--21, 2012.

\bibitem{Bringmann14}
K.~Bringmann.
\newblock Why walking the dog takes time: {F}r\'echet distance has no strongly
  subquadratic algorithms unless {SETH} fails.
\newblock In {\em Proc. 55th Annu. IEEE Sympos. Found. Comput. Sci. (FOCS)},
  pages 661--670, 2014.

\bibitem{BringmannMu16}
K.~Bringmann and W.~Mulzer.
\newblock Approximability of the discrete {F}r{\'{e}}chet distance.
\newblock {\em Journal of Computational Geometry}, 7(2):46--76, 2016.

\bibitem{BuchinBG10}
K.~Buchin, M.~Buchin, and J.~Gudmundsson.
\newblock {Constrained free space diagrams: a tool for trajectory analysis}.
\newblock {\em Int. J. of GIS}, 24(7):1101--1125, 2010.

\bibitem{BuchinBGLL11}
K.~Buchin, M.~Buchin, J.~Gudmundsson, M.~L\"offler, and J.~Luo.
\newblock Detecting commuting patterns by clustering subtrajectories.
\newblock {\em Internat. J. Comput. Geom. Appl.}, 21(3):253--282, 2011.

\bibitem{BuchinBKRW07}
K.~Buchin, M.~Buchin, C.~Knauer, G.~Rote, and C.~Wenk.
\newblock How difficult is it to walk the dog?
\newblock In {\em Proc. 23rd European Workshop Comput. Geom. (EWCG)}, pages
  170--173, 2007.

\bibitem{BuchinBMS12}
K.~Buchin, M.~Buchin, W.~Meulemans, and B.~Speckmann.
\newblock Locally correct {F}r{\'e}chet matchings.
\newblock In {\em Proc. 20th Annu. European Sympos. Algorithms (ESA)}, volume
  7501 of {\em Lecture Notes in Computer Science}, pages 229--240, 2012.

\bibitem{BuchinBuSc10}
K.~Buchin, M.~Buchin, and A.~Schulz.
\newblock {F}r{\'e}chet distance of surfaces: Some simple hard cases.
\newblock In {\em Proc. 18th Annu. European Sympos. Algorithms (ESA)}, volume
  6347 of {\em Lecture Notes in Computer Science}, pages 63--74, 2010.

\bibitem{bblmm-fdrl-16}
K.~Buchin, M.~Buchin, R.~van Leusden, W.~Meulemans, and W.~Mulzer.
\newblock Computing the {F}r{\'{e}}chet distance with a retractable leash.
\newblock {\em Discrete Comput. Geom.}, 56(2):315--336, 2016.

\bibitem{BuchinBW09}
K.~Buchin, M.~Buchin, and Y.~Wang.
\newblock {Exact algorithms for partial curve matching via the {F}r{\'e}chet
  distance}.
\newblock In {\em Proc. 20th Annu. ACM-SIAM Sympos. Discrete Algorithms
  (SODA)}, pages 645--654, 2009.

\bibitem{BuchinBW08}
K.~Buchin, M.~Buchin, and C.~Wenk.
\newblock Computing the {F}r{\'e}chet distance between simple polygons.
\newblock {\em Comput. Geom. Theory Appl.}, 41(1--2):2--20, 2008.

\bibitem{BuchinMu11}
K.~Buchin and W.~Mulzer.
\newblock {D}elaunay triangulations in {$O({\rm sort}(n))$} time and more.
\newblock {\em J. ACM}, 58(2):Art. 6, 27, 2011.

\bibitem{Buchin07}
M.~Buchin.
\newblock {\em On the Computability of the {F}r{\'e}chet Distance Between
  Triangulated Surfaces}.
\newblock PhD thesis, Free University Berlin, Institute of Computer Science,
  2007.

\bibitem{ChambersVELLT10}
E.~Chambers, {\'E}.~de~Verdi{\`e}re, J.~Erickson, S.~Lazard, F.~Lazarus, and
  S.~Thite.
\newblock Homotopic {F}r{\'e}chet distance between curves or, walking your dog
  in the woods in polynomial time.
\newblock {\em Comput. Geom. Theory Appl.}, 43(3):295--311, 2010.

\bibitem{Chan08}
T.~M. Chan.
\newblock All-pairs shortest paths with real weights in ${O}(n^3 / \log n)$
  time.
\newblock {\em Algorithmica}, 50(2):236--243, 2008.

\bibitem{Chan10}
T.~M. Chan.
\newblock More algorithms for all-pairs shortest paths in weighted graphs.
\newblock {\em SIAM J. Comput.}, 39(5):2075--2089, 2010.

\bibitem{ChazelleLe86}
B.~M. Chazelle and D.~T. Lee.
\newblock On a circle placement problem.
\newblock {\em Computing}, 36(1--2):1--16, 1986.

\bibitem{CookDHSW11}
A.~F. Cook, A.~Driemel, S.~Har-Peled, J.~Sherette, and C.~Wenk.
\newblock Computing the {F}r{\'e}chet distance between folded polygons.
\newblock In {\em Proc. 12th Symposium on Algorithms and Data Structures
  (WADS)}, pages 267--278, 2011.

\bibitem{CookWenk10}
A.~F. Cook and C.~Wenk.
\newblock Geodesic {F}r{\'e}chet distance inside a simple polygon.
\newblock {\em ACM Transactions on Algorithms}, 7(1):Art. 9, 2010.

\bibitem{DriemelHarpeled13}
A.~Driemel and S.~Har{-}Peled.
\newblock Jaywalking your dog: Computing the {F}r{\'{e}}chet distance with
  shortcuts.
\newblock {\em SIAM J. Comput.}, 42(5):1830--1866, 2013.

\bibitem{DriemelHW12}
A.~Driemel, S.~Har{-}Peled, and C.~Wenk.
\newblock Approximating the {F}r{\'{e}}chet distance for realistic curves in
  near linear time.
\newblock {\em Discrete Comput. Geom.}, 48(1):94--127, 2012.

\bibitem{EfratGHMM02}
A.~Efrat, L.~Guibas, S.~Har-Peled, J.~Mitchell, and T.~Murali.
\newblock New similarity measures between polylines with applications to
  morphing and polygon sweeping.
\newblock {\em Discrete Comput. Geom.}, 28(4):535--569, 2002.

\bibitem{EiterMannila94}
T.~Eiter and H.~Mannila.
\newblock {Computing Discrete {F}r{\'e}chet Distance}.
\newblock Technical Report CD-TR 94/65, Christian Doppler Laboratory, 1994.

\bibitem{Erickson99}
J.~Erickson.
\newblock Bounds for linear satisfiability problems.
\newblock {\em Chicago J. Theor. Comput. Sci.}, page Article 8, 1999.

\bibitem{Fredman75}
M.~L. Fredman.
\newblock How good is the information theory bound in sorting?
\newblock {\em Theoret. Comput. Sci.}, 1(4):355--361, 1975/76.

\bibitem{FredmanWi93}
M.~L. Fredman and D.~E. Willard.
\newblock Surpassing the information theoretic bound with fusion trees.
\newblock {\em J. Comput. System Sci.}, 47(3):424--436, 1993.

\bibitem{FurstSaSi84}
M.~L. Furst, J.~B. Saxe, and M.~Sipser.
\newblock Parity, circuits, and the polynomial-time hierarchy.
\newblock {\em Mathematical Systems Theory}, 17(1):13--27, 1984.

\bibitem{GajentaanOv95}
A.~Gajentaan and M.~H. Overmars.
\newblock On a class of ${O}(n^2)$ problems in computational geometry.
\newblock {\em Comput. Geom. Theory Appl.}, 5(3):165--185, 1995.

\bibitem{Godau91a}
M.~Godau.
\newblock A natural metric for curves - computing the distance for polygonal
  chains and approximation algorithms.
\newblock In {\em Proc. 8th Sympos. Theoret. Aspects Comput. Sci. (STACS)},
  pages 127--136, 1991.

\bibitem{Godau98}
M.~Godau.
\newblock {\em {On the Complexity of Measuring the Similarity Between Geometric
  Objects in Higher Dimensions}}.
\newblock PhD thesis, Freie Universit{\"a}t Berlin, Germany, 1998.

\bibitem{GudmundssonWolle10}
J.~Gudmundsson and T.~Wolle.
\newblock {Towards automated football analysis: Algorithms and data
  structures}.
\newblock In {\em Proc. 10th Australasian Conf. on Mathematics and Computers in
  Sport}, 2010.

\bibitem{HarPeledNSS12}
S.~Har-Peled, A.~Nayyeri, M.~Salavatipour, and A.~Sidiropoulos.
\newblock How to walk your dog in the mountains with no magic leash.
\newblock In {\em Proc. 28th Annu. ACM Sympos. Comput. Geom. (SoCG)}, pages
  121--130, 2012.

\bibitem{HarPeledRa14}
S.~Har{-}Peled and B.~Raichel.
\newblock The {F}r{\'{e}}chet distance revisited and extended.
\newblock {\em ACM Transactions on Algorithms}, 10(1):3, 2014.

\bibitem{HirschbergChSa79}
D.~S. Hirschberg, A.~K. Chandra, and D.~V. Sarwate.
\newblock Computing connected components on parallel computers.
\newblock {\em Commun. {ACM}}, 22(8):461--464, 1979.

\bibitem{Indyk02}
P.~Indyk.
\newblock Approximate nearest neighbor algorithms for {F}rechet distance via
  product metrics.
\newblock In {\em Proc. 18th Annu. ACM Sympos. Comput. Geom. (SoCG)}, pages
  102--106, 2002.

\bibitem{GronlundPe14}
A.~G. J{\o}rgensen and S.~Pettie.
\newblock Threesomes, degenerates, and love triangles.
\newblock In {\em Proc. 55th Annu. IEEE Sympos. Found. Comput. Sci. (FOCS)},
  pages 621--630, 2014.

\bibitem{katz1997expander}
M.~Katz and M.~Sharir.
\newblock An expander-based approach to geometric optimization.
\newblock {\em SIAM J. Comput.}, 26(5):1384--1408, 1997.

\bibitem{MaheshwariSSZ11}
A.~Maheshwari, J.-R. Sack, K.~Shahbaz, and H.~Zarrabi-Zadeh.
\newblock Fr{\'e}chet distance with speed limits.
\newblock {\em Comput. Geom. Theory Appl.}, 44(2):110--120, 2011.

\bibitem{MaheshwariSSZ11b}
A.~Maheshwari, J.-R. Sack, K.~Shahbaz, and H.~Zarrabi-Zadeh.
\newblock Improved algorithms for partial curve matching.
\newblock In {\em Proc. 19th Annu. European Sympos. Algorithms (ESA)}, volume
  6942 of {\em Lecture Notes in Computer Science}, pages 518--529, 2011.

\bibitem{PreparataSh85}
F.~P. Preparata and M.~I. Shamos.
\newblock {\em Computational geometry. An introduction.}
\newblock Springer-Verlag, 1985.

\bibitem{Patrascu:2010}
M.~P\v{a}tra\c{s}cu.
\newblock Towards polynomial lower bounds for dynamic problems.
\newblock In {\em Proc. 42nd Annu. ACM Sympos. Theory Comput. (STOC)}, pages
  603--610, 2010.

\bibitem{Schoenhage79}
A.~Sch{\"o}nhage.
\newblock On the power of random access machines.
\newblock In {\em Proc. 6th Internat. Colloq. Automata Lang. Program. (ICALP)},
  pages 520--529, 1979.

\bibitem{Schoenhage80}
A.~Sch{\"{o}}nhage.
\newblock Storage modification machines.
\newblock {\em SIAM J. Comput.}, 9(3):490--508, 1980.

\bibitem{SharirAg95}
M.~Sharir and P.~K. Agarwal.
\newblock {\em {Davenport}-{Schinzel} Sequences and Their Geometric
  Applications}.
\newblock Cambridge University Press, 1995.

\bibitem{Tarjan75}
R.~E. Tarjan.
\newblock Efficiency of a good but not linear set union algorithm.
\newblock {\em J. ACM}, 22(2):215--225, 1975.

\bibitem{Thorup02}
M.~Thorup.
\newblock Randomized sorting in ${O}(n \log \log n)$ time and linear space
  using addition, shift, and bit-wise boolean operations.
\newblock {\em J. Algorithms}, 42(2):205--230, 2002.

\bibitem{WenkSP06}
C.~Wenk, R.~Salas, and D.~Pfoser.
\newblock {Addressing the need for map-matching speed: Localizing global
  curve-matching algorithms}.
\newblock In {\em Proc. 18th Int. Conf. on Sci. and Stat. Database Management},
  pages 379--388, 2006.

\end{thebibliography}
\end{document}